\documentclass[12pt]{article}
\usepackage[margin=0.9in]{geometry}
\usepackage{booktabs} 
\usepackage[ruled]{algorithm2e}
\usepackage{amsmath}
\usepackage{amsthm}
\usepackage{caption}
\usepackage{subcaption}
\usepackage{amssymb}
\usepackage{multirow}
\usepackage{tabularray}

\usepackage{thmtools}
\usepackage{thm-restate}

\usepackage{comment}
\usepackage{xspace}
\usepackage[usenames,dvipsnames]{xcolor}
\usepackage[colorlinks=true,citecolor=ForestGreen,linkcolor=ForestGreen,urlcolor=NavyBlue]{hyperref}
\usepackage[capitalize]{cleveref}

\usepackage{tikz}
\usetikzlibrary{arrows,positioning}

\newtheorem{theorem}{Theorem}[section]
\newtheorem{lemma}[theorem]{Lemma}
\newtheorem{proposition}[theorem]{Proposition}
\newtheorem{definition}[theorem]{Definition}
\newtheorem{corollary}[theorem]{Corollary}
\newtheorem*{remark}{Remark}

\newtheorem{example}{Example}[section]

\DeclareMathOperator*{\argmin}{arg\,min}

\SetAlFnt{\small}
\SetAlCapFnt{\small}
\SetAlCapNameFnt{\small}
\SetAlCapHSkip{0pt}
\IncMargin{-\parindent}
\usepackage{natbib}

\usepackage{footnote}
\makesavenoteenv{algorithm}

\newcommand{\BR}{\text{BR}}
\newcommand{\ut}{\underline{t}}
\newcommand{\bt}{\mathbf{t}}
\newcommand{\x}{\mathbf{x}}
\newcommand{\p}{\mathbf{p}}
\newcommand{\w}{\mathbf{w}}
\newcommand{\z}{\mathbf{z}}

\usepackage{comment}
\newcount\Comments
\newcommand{\kibitz}[2]{\ifnum\Comments=1{\color{#1}{#2}}\fi}

\definecolor{english}{rgb}{0.0, 0.5, 0.0}

\setcitestyle{authoryear}

\renewenvironment{abstract}
 {\small
  \begin{center}
  \bfseries \abstractname\vspace{-.5em}\vspace{0pt}
  \end{center}
  \list{}{
    \listparindent 1.5em%
    \setlength{\leftmargin}{8mm}%
    \itemindent    \listparindent
    \setlength{\rightmargin}{\leftmargin}%
  }%
  \item\relax}
 {\endlist}

\author{
  Yannai A. Gonczarowski\\
  Harvard University\\
  \texttt{yannai@gonch.name}
  \and 
  Gary Qiurui Ma\\
  Harvard Univeristy\\
  \texttt{qiurui\char`_ma@g.harvard.edu}
  \and
  David C. Parkes \\
  Harvard University \\
  \texttt{parkes@eecs.harvard.edu}
}

\title{Pricing with Tips in Three-Sided Delivery Platforms\thanks{Gonczarowski gratefully acknowledges research support by the National Science Foundation (NSF-BSF grant No.\ 2343922), Harvard FAS Inequality in America Initiative, and Harvard FAS Dean’s Competitive Fund for Promising Scholarship. The authors thank Assaf Romm for helpful comments and insightful discussions.}}
\date{}

\begin{document}
\maketitle

\begin{abstract}
We model a delivery platform facilitating transactions among three sides: buyers, stores, and couriers. 
In addition to buyers paying store-specific purchase prices and couriers receiving store--buyer-specific delivery compensation from the platform, each buyer has the option to directly tip for delivery from a specific store. An equilibrium consists of prices, compensations, tips, and transactions that clear the market, such that buyers receive deliveries from preferred stores considering the prices and tips they pay, and couriers deliver preferred orders considering the compensations and tips they receive.

We illustrate the role of tips in pricing: Without tips, an equilibrium is only guaranteed to exist when there are at least as many couriers as buyers or stores. In contrast, with tips an equilibrium always exists. From an efficiency perspective, the optimal with-tip equilibrium welfare is always weakly larger than the optimal without-tip equilibrium welfare.
However, we show that even with tips, efficient equilibria may not exist, and calculating the optimal equilibrium welfare is NP-hard. To address these challenges, we identify natural conditions on market structure that ensure the existence of efficient with-tip equilibria and allow these efficient equilibria to be computed in polynomial time.
\end{abstract}

\section{Introduction}

Online delivery platforms such as UberEats, DoorDash, Grubhub, and Instacart are becoming an essential part of modern life. These platforms facilitate \emph{three-sided} transactions: In each transaction, a \emph{buyer} receives food from a \emph{store} via delivery by a \emph{courier}. Buyers pay for these transactions, while stores and couriers are compensated for them.

Each of the aforementioned delivery platforms allows buyers to offer \emph{tips} in addition to any compensation provided by the platform. Unlike ride-sharing platforms, where tips are primarily post-service gratuities, the use of tips on delivery platforms have several distinctive traits:
\begin{enumerate}
    \item Across all four delivery platforms, tips are determined before a delivery takes place, and couriers can observe both the tip amount and the platform's payment before deciding whether to accept an order.\footnote{On UberEats, DoorDash and Grubhub, couriers see the total fare of an order, which includes the tips before deciding whether to accept an order. On Instacart, tips are listed alongside platform compensation before couriers decide whether to accept an order. See \url{https://help.uber.com/en/ubereats/restaurants/article/how-to-add-a-tip?nodeId=5867f9dd-7ca5-4484-af59-6c222d9a8355},\url{https://courier-support.grubhub.com/hc/en-us/articles/20977077466516-Seattle-Tip-Policy} and \url{https://www.instacart.com/company/shoppers/shopper-earnings}. Links visited on 04/01/2025.} 
    \item Tips factor into a courier's decisions to accept or reject an order.\footnote{Some DoorDash couriers have adopted a ``no tip, no trip'' strategy, where they ignore orders without tips for more profitably ones (\url{https://www.newsweek.com/no-tip-no-trip-undelivered-doordash-orders-spark-debate-viral-video-1697322}).
    As another example, after British Columbia, Canada guaranteed minimum wage protection for delivery couriers in 2024 (\url{https://news.gov.bc.ca/releases/2024LBR0011-000900}), UberEats stopped showing couriers in the province the tip amount at the time of accepting the order. (\url{https://www.reddit.com/r/UberEATS/comments/1f5385n/uber_eats_will_no_longer_show_tips_on_offers_in/}). Links visited on 04/18/2025.}
    \item While delivery platforms implement some form of dynamic pricing to promote the supply of couriers,\footnote{UberEats employs surge pricing to pay couriers during busy hours. See \url{https://help.uber.com/ubereats/restaurants/article/higher-delivery-fee?nodeId=4938ca61-e0d6-493a-9384-eb005c2eb6e5,newsroom.uber.com/delivery-at-uber-speed-even-when-its-busy}. DoorDash couriers receive ``Peak Pay" during busy hours. See \url{https://help.doordash.com/dashers/s/article/Peak-Pay?language=en_US}. Grubhub offers fixed incentives like Missions and Special Offers to boost courier earnings through goal-based rewards, but unlike dynamic pricing, these are preset and not responsive to real-time market changes. See \url{https://courier.grubhub.com/courier-pay/}. Instacart adds ``Pay Boost'' to an order if an order remains unaccepted for a while. See \url{https://www.instacart.com/company/shoppers/shopper-earnings}. Links visited on 04/02/2025.} not all orders are delivered in a timely manner and customers need to tip to receive speedy delivery.\footnote{For instance, DoorDash warns customers that not tipping may lead to long wait times (\url{https://www.nytimes.com/2023/11/02/business/doordash-tip-warning.html}), and in extreme cases, untipped orders may go undelivered for hours (\url{https://www.dailydot.com/irl/doordasher-shows-piles-of-non-tipping-orders/}). Links visited on 04/18/2025.}
    \item Among all orders, couriers choose the most attractive ones based on platform compensation, tip size, and delivery cost. This is a common practice known as ``cherry picking''.\footnote{For a concrete example, a Philadelphia delivery courier declines 75\% of orders and hang out in wealthy areas only and claimed to receive as much as 45\$ per hour. (\url{https://nypost.com/2023/05/05/i-deliver-food-i-only-accept-orders-in-wealthy-areas-to-ensure-good-tips/}). Links visited on 04/10/2025.} On UberEats and Instacart, couriers see a range of nearby delivery orders and not just a single order recommended by the platform\footnote{UberEats matches couriers to orders in two ways: 1) Exclusive offers where a courier is assigned to one delivery order exclusively, and 2) TripRadar offers where a courier can choose between multiple orders simultaneously when they are driving low speed or not moving. Instacart couriers can see multiple orders simultaneously in their ``batch list." DoorDash and Grubhub show each courier a single order at a time. See \url{https://www.uber.com/en-AU/blog/introducing-trip-radar/} and \url{https://www.instacart.com/company/shoppers/access-batches}. Links visited on 04/01/2025.}, which further facilitates cherry picking.
    \item Tips make up a substantial portion of courier earnings on food delivery platforms (see \cref{tab:larger_tip_table}, based on data from \citealp{jacobs2024gig}), making them a key factor in couriers' decisions to accept or reject orders.
    \begin{table}[t]
\centering
\begin{tabular}{|ccc|}
\hline
 & Without Tips & With Tips \\
\hline
Ride-Share, CA & \$7.12 & \$9.09 \\
Ride-Share, Outside CA & \$10.64 & \$12.94 \\
\hline
Food-Delivery, CA & \$5.93 & \$13.62 \\
Food-Delivery, Outside CA & \$0.48 & \$9.87 \\
\hline
\end{tabular}
\caption{Average courier net hourly earnings on ride-share and delivery platforms in 2022, taken from \cite{jacobs2024gig}. CA (California) refers to Los Angeles and San Francisco Bay, while Outside CA refers to  Boston, Chicago and Seattle. In California, gig workers are guaranteed to receive 120\% of local minimum wage.
\label{tab:larger_tip_table}}
\end{table}
\end{enumerate}
Overall, tips on delivery platforms are quite different from an optional expressions of appreciation. Alongside platform-set prices, tips are embedded in the incentive structure of buyers and couriers, and play a crucial role in the market-clearing process.

In this paper, we incorporate the above five traits to theoretically study pricing with tips in three-sided delivery platforms. We seek to address the following two questions:
\begin{enumerate}
    \item What are the benefits of allowing buyers to specify tips before delivery, compared with not allowing them to do so?
    \item How should a platform that is interested in maximizing welfare set buyer prices and courier compensations, considering that tips factor in buyers' and couriers' decision making?  
\end{enumerate}

To answer the first question, we formalize an appropriate equilibrium concept with tips, and show its superiority to without-tip equilibrium in terms of existence and welfare. To address the second question, we first show broad impossibility results in general markets, and then contrast these with positive results for markets that exhibit natural structures.

\subsection{Our Results}

In \cref{sec:model}, we introduce a platform economy with \emph{unit-demand} buyers, \emph{unit-capacity} couriers, and \emph{unit-supply} stores.\footnote{Our results extend to buyers with additive valuations and couriers with additive costs.} Each buyer has a \emph{valuation}  for each store, which is the sum of the buyer's value for the item offered by the store and the buyer's value for having the item delivered. Each buyer views couriers interchangeably. Couriers incur store--buyer-specific \emph{delivery costs}. Each store sets a fixed product price, and has a cost of production.

An equilibrium consists of purchase prices, delivery compensations, tips, and an allocation. For each store, there is a store-specific \emph{purchase price} for the item offered by this store, which is the same price for all buyers, i.e., there is no price discrimination for the same product. Each store's purchase price contains its fixed product price. For each buyer--store pair, there is a buyer--store-specific \emph{delivery compensation}, which is the same for all couriers, i.e., there is no wage discrimination for the same delivery. For each store, each buyer may also offer a store-specific \emph{tip}, which is paid to a courier who delivers her order from this store. Each courier may choose a buyer--store pair to deliver for, which leads to the buyer in this buyer--store pair paying the purchase price plus the tip (if any) offered for the store, and the courier receiving the delivery compensation plus this tip. The platform receives purchase prices from buyers, pays delivery compensations to couriers, and also pays stores for the product.
An \emph{allocation} (three-sided matching) specifies the executed set of transactions, in which buyers purchase from stores and receive deliveries from couriers.

A buyer's utility when buying from a store is her valuation minus the purchase price and the tip, if any, offered for delivery from the store. Each courier's utility equals the delivery compensation received for a completed order plus the buyer's tip, if any, minus the courier's delivery cost. The \emph{welfare} of an allocation is the sum of allocated buyers' valuations minus allocated couriers' delivery costs minus allocated stores' costs. The \emph{optimal welfare} is the maximum welfare of all allocations. One potential source of inefficiency arises if a store's product price is higher than the store's cost, as in this case an optimal-welfare allocation may not be achieved. To turn off this source of inefficiency, we set each store's product price to be equal to its cost.

We develop an appropriate equilibrium notion for a platform market. First, a natural requirement is that agents utility-maximize: Each buyer buys from a store for which she has the highest utility considering her valuations, the purchase prices, and the minimum tip that she has to offer to ensure some courier delivers to her from that store. Each courier delivers an order for which she has the highest utility considering her cost, the delivery compensation, and the tip offered (if any). A second requirement, as is standard in the definition of Walrasian equilibrium in two sided-markets, is that for each store that does not sell, its purchase price and all the tips associated with delivery from the store are zero; and for each buyer--store pair that is not delivered, its delivery compensation is zero.

Rather than budget balance, as has been considered in the multilateral matching-with-contracts literature \citep{hatfield2011multilateral,ostrovsky2008stability}, we model the platform as being able to subsidize some of the delivery cost, allowing a delivery compensation to be larger than the purchase price minus the product price. For example, Instacart and DoorDash both use ``Pay Boost" to guarantee delivery when an order is not accepted for a long time, where the boost is paid for by the delivery platforms and not charged to the buyer. The ``Pay Boost" on Instacart can be as high as \$12, while the base-pay for a courier to complete an order is around \$4.\footnote{See reports from \url{https://teachmedelivery.com/courier/instacart-peak-boost/}, \url{https://www.fastcompany.com/90305854/exclusive-doordash-reveals-how-much-it-relies-on-customer-tips-to-pay-its-workers} and \url{https://help.doordash.com/dashers/s/article/How-is-Dasher-pay-calculated?}. Links visited on 04/02/2025.} Indeed, each of UberEats, DoorDash and Instacart only recently turned a profit and have, in effect, been subsidizing the economic activity on their platforms, while Grubhub are still reporting losses.\footnote{Doordash generated positive net income for the first time as a public company in Q3, 2024 (\url{https://ir.doordash.com/news/news-details/2024/DoorDash-Releases-Third-Quarter-2024-Financial-Results/default.aspx}); Instacart reached positive net income in Q4, 2023 (\url{https://investors.instacart.com/static-files/45c59490-b81c-4401-ba22-be639847baa7}); UberEats is a subsidiary of Uber, which itself turned profitable in Q2, 2023 (\url{https://investor.uber.com/news-events/news/press-release-details/2023/Uber-Announces-Results-for-Second-Quarter-2023/default.aspx}). Grubhub reported net losses for both 2023 and 2024 (\url{https://electroiq.com/stats/grubhub-statistics}). Links visited on 02/02/2025.}

In \cref{sec:motivate_use_tip}, we explore the benefits of allowing buyers to specify tips before delivery, as opposed to not allowing tips or only allowing tips after delivery. We refer to an equilibrium in the first setting as a "with-tip equilibrium" and to an equilibrium in either of the latter two settings as a "without-tip equilibrium." (Note that from the perspective of the incentives of risk-neutral couriers, not allowing tips is equivalent to allowing them only after delivery.) 

Our first result is a dichotomy: Without-tip equilibria are only guaranteed to exist when there are at least as many couriers as buyers or stores, while with-tip equilibria always exist. This is because in the without-tip case, courier shortages leave some buyers without delivery despite willingness to pay; with tips, the platform can set high compensations for some buyer--store pairs, making others require prohibitively high tips for delivery. Secondly, every allocation supported in a without-tip equilibrium is also supported in a with-tip equilibrium, implying that the optimal with-tip equilibrium welfare is always weakly higher than optimal without-tip equilibrium welfare.
We illustrate this with a market instance where all without-tip equilibria are highly inefficient, yet there exists a with-tip equilibrium that achieves the optimal welfare. 

Perhaps more surprisingly, the welfare gain of with-tip equilibria arises not from the actual (on-path) payment of tips, but rather from the off-path need for buyers to offer them to secure delivery. In fact, every with-tip equilibrium allocation can be supported by a with-tip equilibrium in which all tips are zero---even on realized deliveries---highlighting that tips serve primarily to prevent off-path deviations rather than facilitate on-path transfers. The platform can subsidize delivery for buyer--store pairs in equilibrium, allowing buyers to pay zero tips for the equilibrium allocation, while requiring prohibitively high tips to attain delivery outside of the equilibrium allocation.
 
While tips guarantee the existence of equilibria and improve welfare, we further show in \cref{sec:general_markets} that even with-tip equilibria can, in general, lack desirable properties. First, there exist markets where every with-tip equilibrium is inefficient. 
Second, it is NP-hard to calculate the optimal equilibrium welfare. In fact, we show that it is NP-hard to calculate the optimal welfare, regardless of equilibrium requirements. These limitations raise the question as to whether there are natural structural assumptions under which these economic and computational impossibilities can be circumvented. Our main result is a positive answer to this question.

In \cref{sec:market_structures}, we identify two natural structural assumptions that overcome these impossibilities: either delivery costs are decomposable into courier--store and store--buyer components, which reflects realistic distance-based costs that account for couriers' trips to stores and then trips to buyers; or delivery costs are  decomposable into store--buyer and buyer--courier components, which reflects distance-based costs that account for couriers' trips from stores to buyers, and couriers' return-home trips.
The courier--store component can also account for a courier's familiarity with a store, where grocery pickups can be time-consuming for unfamiliar locations.
In markets that satisfy one of these two delivery cost structures, 
there always exists an efficient with-tip equilibrium that can be found in polynomial time. In contrast, there exist markets with each of these structures in which without-tip equilibria do not exists, and if exist have arbitrarily low optimal without-tip equilibrium welfare. This underscores the importance of tips, even in these settings with naturally structured delivery costs.

In \cref{sec:market_structures_buyers}, we further consider structures of buyer valuations in place of structures of delivery costs. When buyers are {\em single-minded} (i.e., each buyer has positive valuation for just one store),
we show there exists an efficient with-tip equilibrium that can be found in polynomial time regardless of delivery cost structures. However, just as in the case with delivery cost structures, without-tip equilibria are not guaranteed to exist, and if exist can have low welfare. In \cref{sec:dis}, we discuss some limitations of our model and suggest directions for future work.

\subsection{Related Work}

Our work joins several strands of research, including the analysis of competitive equilibrium on trading networks \citep{ostrovsky2008stability,hatfield2013stability} and multilateral contracting \citep{hatfield2011multilateral}. These works extend results on matching and equilibria in two-sided markets \cite{shapley1971assignment,kelso1982job,gul1999walrasian} to richer environments, proving competitive or Walrasian equilibria are efficient and exist with suitable conditions on valuations. However, all  these works either require bilateral contracts, or permit agents to engage in fractional participation within a multilateral contract. In contrast, the matching on delivery platforms is discrete and between three sides.

For discrete, three-sided matching, competitive equilibrium may not exist \citep{alkan1988nonexistence}, and we further show in \cref{app:ce_existence} that determining the existence of a competitive equilibrium is NP-hard in our setting.\footnote{See also \cref{app:ce_existence} for a discussion as to whether, and to what extent, each of the first- and second-welfare theorems still hold for competitive equilibria in 3-sided markets.} We therefore move away from a three-sided competitive equilibrium model,  instead allowing the platform to subsidize delivery and abstracting away stores' incentives while still keeping the stores' capacity constraints.
With this relaxed notion of equilibrium---which, together with the availability of tips, restores the existence of equilibrium but loses the first welfare theorem---we ask how much welfare can be attained.

Other works in the envy-free pricing literature \citep{guruswami2005profit,chen2010envy} have studied how a seller in a two-sided market can set prices to maximize revenue, including under various supply constraints \citep{cheung2008approximation,im2012envy}. 
As with these works, we require the equilibrium allocation to be envy free for buyers given purchase prices. However, unlike supply constraints in a two-sided market, and hence unlike all of the above works, our three-sided problem also requires envy freeness for couriers, with buyers’ choices needing to align with couriers’ delivery decisions. Additionally, while these earlier works assume a single seller that sets prices and allow for nonzero prices for unsold items, our model differs in requiring unsold stores to have zero purchase prices and orders not delivered to have zero compensations.

We establish the computational hardness of finding an optimal-welfare allocation in a three-sided delivery platform by reducing from the maximum-weight 3-dimensional matching problem. While some works find a matching that achieves half of the optimal welfare in polynomial time in this problem~\citep{arkin1998local,halldorsson1995approximating,chan2009linear}, they rely on local search algorithms that swap subsets of individual matches without considering envy-free requirements, making them unsuitable for our setting. The NP-hardness of finding an optimal-welfare allocation even without equilibrium requirements precludes the possibility of devising an algorithm for finding an efficient, i.e., welfare-optimal,
equilibrium in our setting. Therefore, we focus on identifying suitable structural assumptions on courier costs or on buyer valuations that limit couriers' and buyers' envy, and allow for the existence of equilibria that achieve the optimal welfare.

Previous theoretical works analyze tips paid after service completion and seek to rationalize tipping through repeated interactions \citep{ben1976tip}, social norms \citep{azar2005social,azar2007pay,debo2018tipping,snitkovsky2021modeling}, and altruism \citep{lynn2015service}. In contrast, works analyzing tips before service are scarce. \citet{lei2023two} suggest that delivery platforms may use upfront tips to reduce courier wages under competitive pressure. Our work provides a new perspective, showing that equilibria with tips before service weakly Pareto dominate those without tips or after service. Since buyers in our model account for both the utility impact of tipping and which store they are served by, we offer a novel rationale for why tipping—traditionally a post-service gesture—is specified upfront on delivery platforms. There is also a body of empirical literature on tipping, including studies on how default tip suggestions affect buyer satisfaction \citep{alexander2021effects, haggag2014default} and courier welfare \citep{shy2015tips,castillo2022designing}.

This study also contributes to the broader literature, on how online platforms set prices and facilitate matching, in three-sided markets (\citep{wang2025recommending,liu2023operating,bahrami2023three}), as well as two-sided ones \citep{wang2023platform,eden2023platform,d2024disrupting,birge2021optimal,ma2020spatio,feldman2023managing,chen2022food,chen2024courier}. While we share the common feature of a central platform using pricing and matching tools for market design, our work focuses on a different market dimension--- the role of tips.

\section{Model}\label{sec:model}

A delivery platform facilitates transactions between three sides of a market. There is a set of $m$ unit-demand buyers, $B=\{b_1,...,b_m\}$, $n$ unit-supply stores $S=\{s_1,...,s_n\}$, and $l$ unit-capacity couriers, $D=\{d_1,...,d_l\}$.\footnote{All our results generalize beyond unit-demand to buyers with additive valuations, and beyond unit-capacity to couriers with additive costs. For buyers with additive valuations, duplicate each buyer $n$ times, where the $k$-th duplicate of buyer $i$ (denoted as $v_{b_{i,k}}$) has valuation $v_{b_{i,k}}(s_k)=v_{b_i}(s_k)$ and $v_{b_{i,k}}(s'_k)=0$ for $s'_k\neq s_k$. For couriers with additive costs, duplicate each courier $mn$ times, each duplicate having unit capacity to serve one buyer--store pair, and infinite cost to serve any other buyer--store pairs.} An {\em order} (from a store, not a mathematical order) $o=(b,s)\in B\times S$ is a buyer--store pair, and the set of all orders is denoted as $O=B\times S$. 
Buyer $b$ has {\em valuation} $v_{b}(s)\geq 0$ for store $s\in S$, which includes both the buyer's valuation for the item of store $s$, and the buyer's valuation for the delivery. Courier $d$ incurs finite {\em delivery cost} $c_{d}(b,s)\geq 0$ for delivering from store $s$ to buyer $b$.
Each store $s$ charges the platform a predetermined product price when a buyer buys from it through the platform. We assume this price equals its cost.\footnote{We model product prices as being equal to stores' costs to 
remove the following source of inefficiency: If a buyer's valuation is smaller than the product price but larger than cost, a trade cannot occur despite the positive welfare generated. In this case, the optimal welfare cannot be obtained even without requiring trades occur in an equilibrium.} 
Without loss of generality, we further set each store's cost (and hence the product price) to zero, and show how to generalize all of our results to the case of nonzero store costs, with product prices still equaling store costs, in \cref{app:zero_store_cost}. 

We now  introduce some of the components of an equilibrium, which include purchase prices, delivery compensation, and later we will also introduce tips.
For each store $s\in S$, there is a {\em purchase price}, $p_{s}\geq 0$, which includes both the predetermined product price (assumed equal to a store's cost, which is set to zero) and the platform delivery fee. We require the purchase price for a store being the same for all buyers, i.e., no price discrimination. We abstract away from discriminating purchase prices based on a buyer's distance from the store, as such distance-based pricing is often coarse; for example, Instacart applies a long-distance service fee only if the delivery route exceeds 30 minutes or includes a toll.\footnote{See \url{https://www.instacart.com/help/section/360007902791/360039164252}. Visited on 04/09/2025.} That said, platforms do charge stores varying fees for fulfilling orders to buyers at different distances,\footnote{DoorDash charges restaurants and partners who use the DoorDash DriveAPI distance-based fees for delivery, but these fees are not charged to buyers. See \url{https://help.doordash.com/consumers/s/article/What-is-a-Delivery-Radius-or-a-Delivery-Area-on-DoorDash?language=en_US}, \url{https://developer.doordash.com/en-US/docs/drive/overview/about_drive}. On Grubhub and UberEats, delivery fees for buyers are chosen by individual stores. Stores choose from Basic, Plus, or Premium plans, each of which comes with different delivery fee levels for buyers. (\url{https://get.grubhub.com/grubhub-pricing-and-fees/}, \url{https://merchants.ubereats.com/us/en/pricing/}). Links visited on 04/09/2025.} but these fees are typically not charged to buyers.

There is also a {\em delivery compensation}, $w_{b s}\geq 0$, for each order $(b,s)$, the same for every courier, i.e., there is no wage discrimination for the same delivery. Let $\mathbf{p}\in R^m$ and $\mathbf{w}\in R^{mn}$ denote the vector of purchase price and delivery compensation. The platform charges purchase prices from buyers and pays delivery compensations to couriers. The platform can
subsidize some of the cost of delivery, allowing delivery compensations $w_{bs}$ to a courier to be larger than purchase price $p_s$ collected
from a buyer. This reflects the ``Pay Boost" adopted by some delivery platforms to guarantee delivery when an order is not accepted for a long time. For example, pay boost on Instacart can be as high as \$12 when the base pay per order for a courier is only around \$4.\footnote{See Footnote~8.} The source of subsidy $w_{bs}-p_s$ can come from membership fees or other fixed fees that are not charged per order.

\paragraph{Without-tip.} The utility associated with a trade $(b,s,d)$ is as follows: buyer $b$ has utility $u_b(s)=v_b(s)-p_s$; courier $d$ has utility $u_d(b, s)=w_{bs}-c_d(b,s)$; the platform has utility that equals the purchase price $p_s$ minus the courier compensation $p_{bs}$; a store has utility zero. The welfare created by a trade is the sum of the four utilities, which adds up to $v_b(s)-c_d(b, s)$.

An allocation, or a three-sided matching $\mathbf{x}$, is defined as  
\begin{equation*}
    x_{bsd} =
    \begin{cases}
      1 & \text{if $b$ buys from $s$, and $d$ serves the delivery,}\\
      0 & \text{otherwise.}
    \end{cases}       
\end{equation*}

An allocation  $\mathbf{x}$ is feasible if it satisfies unit-demand $\forall b, \sum_{sd}x_{bsd}\leq 1$, unit-supply $\forall s, \sum_{bd}x_{bsd}\leq 1$ and unit-capacity $\forall d, \sum_{bs}x_{bsd}\leq 1$. For an allocation $\mathbf{x}$, let $s_\mathbf{x}(b)$ be the store buyer $b$ buys from, $d_\mathbf{x}(b)$ be the courier that delivers for buyer $b$, and $o_\mathbf{x}(d)=(b_\mathbf{x}(d),s_\mathbf{x}(d))$ the buyer--store pair that courier $d$ delivers.
If $b$ does not purchase, let $s_\mathbf{x}(b)=\emptyset$ be a null store, and $d_\mathbf{x}(b)=\emptyset$ be a null courier. If $d$ does not deliver, then let $o_\mathbf{x}(d)=\emptyset$ be a null buyer--store pair. When the context is clear, we use $u_b(\mathbf{x})=u_b(s_\mathbf{x}(b)),u_d(\mathbf{x})=u_d(o_\mathbf{x}(d))$ as short hand for the utility of a buyer and courier, respectively, in allocation $\mathbf{x}$. A buyer who does not buy and a courier who does not deliver have zero utility. All values, costs and utilities associated with null are zero. 

The welfare of an allocation $\mathbf{x}$ is the sum of the welfare of all realized trades, and equals the sum of allocated buyers’ valuations minus the sum of allocated couriers' delivery costs: 
$$W(\mathbf{x})=\sum_{bsd}x_{bsd}(v_b(s)-c_d(b, s)).$$ 
The optimal welfare is the maximum welfare over all feasible allocations, $\mathit{OPT}=\max_{\mathbf{x}}W(\mathbf{x})$.

Given delivery compensation $\mathbf{w}$, denote the set of orders that maximize courier $d$'s utility as
\begin{equation*}
    \BR_d(\mathbf{w})=
    \begin{cases}
      \{(b,s)\;|\; \forall (b',s')\in O, u_d(b,s)\geq u_d(b',s')\} & \text{ if } \exists (b,s) \text{ such that } u_d(b,s)>0\\
      \{\emptyset\}\cup \{(b,s)\;|\; u_d(b,s)=0\} & \text{otherwise.}
    \end{cases}       
\end{equation*}
Given purchase price $\mathbf{p}$, denote the set of stores that maximize buyer $b$'s utility as 
\begin{equation*}
    \BR_b(\mathbf{p})=
    \begin{cases}
      \{s\;|\; \forall s'\in S, u_b(s)\geq u_b(s')\} & \text{ if } \exists s \text{ such that } u_b(s)>0\\
      \{\emptyset\}\cup \{s\;|\; u_b(s)=0\} & \text{otherwise.}
    \end{cases}       
\end{equation*}

\begin{definition}[Without-tip equilibrium]\label{def:withouttipequil}
    A without-tip equilibrium is a triple of prices, compensations, and feasible allocation $(\mathbf{p},\mathbf{w},\mathbf{x})$ that satisfies \begin{itemize} 
    \item Every buyer buys from the favorite store $\forall b, \; s_\mathbf{x}(b)\in \BR_b(\mathbf{p})$.
    \item Every courier delivers the favorite order $\forall d, \; o_\mathbf{x}(d) \in \BR_d(\mathbf{w})$.
    \item Stores not bought from (i.e., $\sum_{bd}x_{bsd}=0$) have zero purchase price $p_s=0$.
    \item Order not delivered (i.e., $\sum_d x_{bsd}=0$) have zero delivery compensation $w_{bs}=0$.
\end{itemize}
\end{definition}
From the perspective of courier decision-making, the without-tip equilibrium also captures scenarios in which buyers specify tips after delivery, as such tips would have no or minimal influence on courier behavior.\footnote{One may consider the case where risk-neutral couriers hold a calibrated belief, that $\alpha$ percentage of purchase price is specified as tip after delivery. Couriers might favor stores with higher purchase price in the hope of a higher tip. This case can also be captured by a without-tip equilibrium, where purchase price for buyers are increased by $\alpha$ percent, and courier compensation for any store $s$ is also increase by $\alpha p_s$.}

The equilibrium definition assumes that couriers can observe all available orders, whereas real-world platforms typically use centralized matching systems that limit courier choices. We justify this assumption with two reasons. First, on platforms like UberEats and Instacart, couriers often see a range of nearby orders rather than just a single platform-recommended option.\footnote{See Footnote~6.} 
In most cases, the order that maximizes a courier's utility is one that begins close to the courier's current location. It is reasonable therefore, that we model couriers as choosing the choices among all choices that maximizes the utility. Second, modeling couriers as facing competitive prices across all orders guarantees that they are envy-free and always accept the platform-dispatched order, assuming the platform assigns each courier their most preferred available order.

It will also be useful to define a \emph{Walrasian equilibrium} for the two-sided market of buyers and stores $(B,S)$, which completely ignores couriers. We will use this concept to analyze couriers' and buyers' incentives in a three-sided allocation in \cref{sec:eq_existence_sec} and \cref{sec:market_structures}.
A \emph{buyer allocation} $\mathbf{z}$ is a two-sided matching in $(B,S)$ where 
\begin{equation*}
    z_{bs} =
    \begin{cases}
      1 & \text{if $b$ buys from store $s$,}\\
      0 & \text{otherwise.}
    \end{cases}       
\end{equation*}

A feasible buyer allocation $\mathbf{z}$ satisfies unit-demand buyer $\forall b, \sum_s z_{bs}\leq 1$ and unit-supply store $\forall s, \sum_b z_{bs}\leq 1$. Let $s_\mathbf{z}(b)$ be the store that buyer $b$ buys from, and $s_\mathbf{z}(b)=\emptyset$ if $b$ does not buy. A Walrasian equilibrium for $(B,S)$ is a pair of purchase price and buyer allocation $(\mathbf{p},\mathbf{z})$, such that all buyers buy from the store that maximizes buyers' utility $s_\mathbf{z}(x)\in \BR_b(\mathbf{p})$, with purchase prices set to zero for stores that do not sell. A Walrasian equilibrium $(\mathbf{p},\mathbf{z})$ always exists in this unit-demand-supply setting \citep{gul1999walrasian}.\footnote{Although our definition of a without-tip equilibrium resembles a Walrasian equilibrium in two-sided markets, we do not use the term Walrasian since our model abstracts away the stores' incentives, and allows the platform to subsidize delivery. See \cref{app:ce_existence} for a discussion of store incentives and \cref{sec:dis} for a discussion of the assumption that the platform can subsidize delivery.}

\paragraph{With-tip.} When buyers are allowed to tip, let $t_{bs}\geq 0$ be the tip associated with an order $o=(b,s)$. Buyer $b$ pays $t_{bs}$ to the courier who delivers the order $(b,s)$, but does not pay out tips for orders that are not delivered. Let $\mathbf{t}\in R^{mn}$ be the vector of tips, $\mathbf{t}_b\in R^{n}$ be the vector of tips associated with buyer $b$, and $\mathbf{t}_{-b}\in R^{(m-1)n}$ the tips associated with all buyers except $b$. Buyer $b$ buying from store $s$ with tip $t_{bs}$ has utility $u_b(s)=v_b(s)-p_s-t_{bs}$.
courier $d$ delivering an order $(b,s)$ has utility $u_d(b,s)=p_{bs}-c_d(b,s)+t_{bs}$. The utility of the platform and stores, as well as the definition of an allocation and welfare remain the same as in the without-tip setting. 

Given the set of delivery compensations $\mathbf{w}$ and tips $\mathbf{t}$, the set of orders that maximize courier $d$'s utility is still defined as
\begin{equation*}
    \BR_d(\mathbf{w},\mathbf{t})=
    \begin{cases}
      \{(b,s)\;|\; \forall (b',s')\in O, u_d(b,s)\geq u_d(b',s')\} & \text{ if } \exists (b,s) \text{ such that } u_d(b,s)>0\\
      \{\emptyset\}\cup \{(b,s)\;|\; u_d(b,s)=0\} & \text{otherwise.}
    \end{cases}       
\end{equation*}

Since in online platforms buyers typically decide the tip amount, defining the set of stores that maximize a buyer's utility is more complex. Let a buyer $b$ consider buying from a store $s$. Given compensations $\mathbf{w}$ and all other buyers' tips $\mathbf{t}_{-b}$, 
there is a minimum tip required to have some courier deliver from the store $s$ to $b$, denoted as $$\underline{t}_{bs}(\mathbf{w},\mathbf{t}_{-b}) = \min\{t_{bs} \;|\; \exists d \text{ such that } (b,s)\in \BR_d(\mathbf{w},\mathbf{t}_b,\mathbf{t}_{-b}).\}$$

This minimum tip $\underline{t}_{bs}(\mathbf{w},\mathbf{t}_{-b})$ has a closed form solution, which we present in \cref{app:min_tip}. For simplicity of notation, we omit the dependency on $\mathbf{w},\mathbf{t}_{-b}$ and write $\underline{t}_{bs}$ from now on. The buyer can calculate the minimum tip $\underline{t}_{bs'}$ for any store $s'\in S$. 
Given this, the set of stores that maximize buyer $b$'s utility is defined as
\begin{equation*}
    \BR_b(\mathbf{p},\mathbf{w},\mathbf{t}_{-b})=
    \begin{cases}
        \{s\;|\;\forall s'\in S, v_b(s)-p_s-\underline{t}_{bs}\geq v_b(s')-p_{s'}-\underline{t}_{bs'} \} & \text{ if } \exists s \text{ such that } v_b(s)-p_s-\underline{t}_{bs}> 0\\
        \{\emptyset\}\cup \{s\;|\; v_b(s)-p_s-\underline{t}_{bs}=0\} & \text{otherwise.}
    \end{cases}
\end{equation*}
The above formulation allows buyers to set tips and consider the minimum tip required for delivery when deciding which store maximizes their utility. This can reflect real-world behavior. For example, in the special case of a buyer who urgently needs a particular item, the buyer may offer a tip just high enough to ensure prompt delivery, but no higher than necessary. By the definition, a buyer $b$ views $v_b(s)-p_s-\underline{t}_{bs}$ as the  utility of buying from a store $s$, and similarly $v_b(s')-p_{s'}-\underline{t}_{bs'}$ as the  utility of buying from store $s'$. The model of utility used in this definition implies every buyer is optimistic in believing that she can get the item even if there are other buyers who contend for the same store $s$.

\begin{definition}[With-tip equilibrium]\label{def:with_tip_eq}
A with-tip equilibrium is a quadruple of prices, compensations, tips, and feasible allocation $(\mathbf{p},\mathbf{w},\mathbf{t},\mathbf{x})$ that satisfies
\begin{itemize}
    \item Every buyer buys from the favorite store $\forall b, \; s_\mathbf{x}(b)\in \BR_b(\mathbf{p},\mathbf{w},\mathbf{t}_{-b})$; and pays the minimum tip if buying from a store, $\forall b, s_\mathbf{x}(b)\neq \emptyset, t_{bs_\mathbf{x}(b)}=\underline{t}_{bs_\mathbf{x}(b)}$.
    \item Every courier delivers the favorite order $\forall d, \; o_\mathbf{x}(d) \in \BR_d(\mathbf{w},\mathbf{t})$.
    \item Stores not bought from (i.e., $\sum_{bd}x_{bsd}=0$) have zero purchase price $p_s=0$.
    \item Order not delivered (i.e., $\sum_d x_{bsd}=0$) have zero delivery compensation $w_{bs}=0$ and zero tips $t_{bs}$=0.
\end{itemize}
\end{definition}

This definition of a with-tip equilibrium allows buyers to reason in a counterfactual way. Given tip vector $\mathbf{t}_{-b}$, a buyer $b$ calculates the minimum tip required to attain delivery for each store. Although formulated so that buyers can tip delivery differently based on different stores, in equilibrium each buyer specifies at most one non-zero tip. We also stress here that the tips are not the only mechanism through which the market clears. Instead, the tips offered by buyers, if any, work alongside the purchase prices and delivery compensations.

\begin{remark} 
An alternative, simpler definition of a with-tip equilibrium, which does not capture the ability of  buyers to set the tips is: let $u_b(s)=v_b(s)-p_s-t_{bs}$, and define
\begin{equation*}
    \BR_b(\mathbf{p},\mathbf{w},\mathbf{t})=
    \begin{cases}
        \{s\;|\;\forall s'\in S, u_b(s)\geq u_b(s')\} & \text{ if } \exists s \text{ such that } u_b(s)> 0\\
        \{\emptyset\}\cup \{s\;|\; u_b(s)=0\} & \text{otherwise.}
    \end{cases}
\end{equation*}

The definition would otherwise remain the same as \cref{def:with_tip_eq}. In this alternative equilibrium definition, tips are a given amount that buyers must pay to buy from a store, rather than the minimum amount that a buyer could offer in incentivizing delivery. As such, this alternative is not economically sensible, removing buyer agency in setting tips. 
\end{remark}

We now illustrate Definition~\ref{def:withouttipequil}  and Definition~\ref{def:with_tip_eq} of with-tip and without-tip equilibria, respectively.
\begin{example}\label{example:without_tip_bad}
    Consider a simple market in Figure~\ref{fig:without_tip_bad}. Buyers $b_1$ and $b_2$ value the store $s_1$ at $3$ and $10$, respectively. Couriers $d_1$ and $d_2$ have delivery costs $c_{d_1}(b_1,s_1)=0,c_{d_1}(b_2,s_1)=11$ and $c_{d_2}(b_1,s_1)=1,c_{d_2}(b_2,s_1)=12$. The optimal welfare is $3$ with the allocation in which $d_1$ delivers from $s_1$ to $b_1$. 
    
    In the without-tip regime, the optimal without-tip equilibrium welfare is $-1$. The equilibrium is defined by $p_{s_1}\in[3,10], w_{b_2s_1}=11,w_{b_1s_1}=0$, and $d_1$ delivers to $b_2$. To see this, there is no equilibrium where $p_{s_1}<3$ because both $b_1$ and $b_2$ want to buy from $s_1$. If $p_{s_1}>10$, the market does not clear.

    In the with-tip equilibrium, the optimal with-tip equilibrium welfare is $3$. One such equilibrium is given by $p_{s_1}=1, w_{b_1s_1}=1, w_{b_2s_1}=0, t_{b_1s_1}=t_{b_2,s_1}=0$ and allocation $d_1$ delivers from $s_1$ to $b_1$. To verify, couriers best responses are $\BR_{d_1}(\mathbf{w},\mathbf{t})=\{(b_1,s_1)\}, \BR_{d_2}(\mathbf{w},\mathbf{t})=\{\emptyset,(b_1,s_1)\}$. Buyer $b_1$ has positive utility $u_{b_1}(s_1)=2>0$, and offers zero tip. Buyer $b_2$ has to tip at least 12 for courier $d_2$ and 12 for courier $d_1$ to deliver from $s_1$ to her. Observe that $\underline{t}_{b_2s_1}=12$, so if deviating, buyer $b_2$ has to pay $p_{s_1}+\underline{t}_{b_2s_1}=13>v_{b_2s_1}$.

    There is another with-tip equilibrium that shares the same purchase prices, delivery compensations, and allocations as the without-tip equilibrium, but with zero tips. To check, couriers incentives remain unchanged with zero tips, $\BR_{d_1}(\mathbf{w},0)=\BR_{d_1}(\mathbf{w},0)=\{\emptyset, (b_2,s_1)\}$. For $b_1$ the minimum tip to have courier $d_1$ deliver for her is $\underline{t}_{b_1s_1}=0$. But $v_{b_1}(s_1)-p_{s_1}-\underline{t}_{b_1s_1}\leq 0$. For $b_2$ it holds that $\underline{t}_{b_2s_1}=0$.
\end{example}
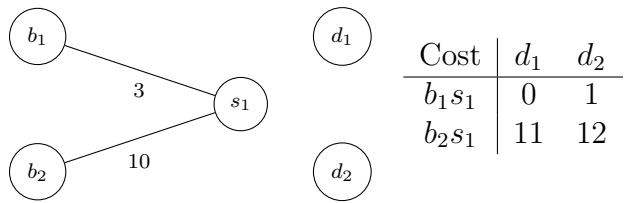
\begin{figure}[t]
    \centering
    \begin{tblr}{}
    \centering
\begin{tikzpicture}[baseline=(current bounding box.center),scale=0.9]
    \node[draw, shape=circle, minimum size=0.1cm] (1) at (-1, 2) {\scriptsize $b_1$};
    \node[draw, shape=circle, minimum size=0.1cm] (2) at (-1, 0) {\scriptsize $b_2$};
    \node[draw, shape=circle, minimum size=0.1cm] (3) at (2, 1) {\scriptsize $s_1$};
    \node[draw, shape=circle, minimum size=0.1cm] (4) at (3.5, 2) {\scriptsize $d_1$};
    \node[draw, shape=circle, minimum size=0.1cm] (5) at (3.5, 0) {\scriptsize $d_2$};
    \draw (1) to (3);
    \draw (2) to (3);
    \node[] at (0.5, 1.2){\scriptsize 3};
    \node[] at (0.5, 0.15){\scriptsize 10};
\end{tikzpicture}
& 
\begin{tabular}{c|cc}
Cost & $d_1$  & $d_2$\\
\hline
$b_1 s_1$ & $0$  & $1$\\
$b_2 s_1$ & $11$  & $12$\\
\end{tabular}
\end{tblr}
    \caption{A market where the optimal without-tip equilibrium welfare is -1, but a with-tip equilibrium achieves the optimal welfare. Buyers $b_1$ and $b_2$ value the only store, $s_1$, at $3$ and $10$, respectively. Courier $d_1$ has costs of delivery $0$ and $11$ to $b_1$ and $b_2$, respectively; courier $d_2$ has costs of delivery $1$ and $12$ to $b_1$ and $b_2$, respectively.
    \label{fig:without_tip_bad}}
\end{figure}

\section{Motivating the Use of Tips}\label{sec:motivate_use_tip}

In this section, we highlight the benefits of incorporating tips into pricing. We show in \cref{sec:eq_existence_sec} that while with-tip equilibria always exist, the existence of without-tip equilibria require a sufficient number of couriers.
In \cref{sec:role_tips}, we show every allocation in a without-tip equilibrium is also in a with-tip equilibrium, meaning the existence of a with-tip equilibrium with weakly larger welfare than all without--tip equilibria.
Lastly, we demonstrate that when tipping is allowed, every equilibrium allocation is also in an equilibrium where all tips equal zero, pinpointing the role of tips as preventing deviations off path rather than for on-path transfers. Most proofs are relegated to \cref{app:eq_existence}.

\subsection{Equilibrium Existence}\label{sec:eq_existence_sec}

We first examine the existence of equilibrium in the without-tip and with-tip regime.
\begin{theorem}\label{thm:eq_existence}
    There always exists a with-tip equilibrium. In contrast, a without-tip equilibrium is only guaranteed to exist when the number of couriers is weakly larger than the minimum number of buyers and stores $l\geq \min\{m,n\}$. There is a market where $l> \min\{m,n\}$ and no without-tip equilibrium exists.
\end{theorem}

\cref{thm:eq_existence} is the combined result of \cref{lem:exist_courier_plan_serve}, \cref{thm:without_tip_eq_existence} and \cref{thm:with_tip_eq_existence}.
As a high-level intuition, a without-tip equilibrium requires that the two-sided market of buyers and stores forms a Walrasian equilibrium. When $l<\min\{m,n\}$, there can be insufficient number of couriers to deliver all orders in this walrasian equilibrium; e.g., the market in \cref{fig:without_tip_not_exists} has $l=1<m=n=2$, and no without-tip equilibrium exists. To explain it in another way, courier shortage leads to some buyers unable to secure delivery despite their high willingness to pay. However, in the with-tip regime platform can heavily subsidize delivery, so that buyers will have to tip a very high amount should they buy from a store that does not sell.
\begin{figure}[t]
    \centering
    \begin{tblr}{}
    \centering
\begin{tikzpicture}[baseline=(current bounding box.center),scale=0.9]
    \node[draw, shape=circle, minimum size=0.1cm] (1) at (-1, 2) {\scriptsize $b_1$};
    \node[draw, shape=circle, minimum size=0.1cm] (2) at (-1, 0) {\scriptsize $b_2$};
    \node[draw, shape=circle, minimum size=0.1cm] (3) at (2, 2) {\scriptsize $s_1$};
    \node[draw, shape=circle, minimum size=0.1cm] (4) at (2, 0) {\scriptsize $s_2$};
    \node[draw, shape=circle, minimum size=0.1cm] (5) at (3.5, 1) {\scriptsize $d_1$};
    \draw (1) to (3);
    \draw (1) to (4);
    \draw (2) to (3);
    \draw (2) to (4);
    \node[] at (0.5, 2.2){\scriptsize 4};
    \node[] at (-0.5, 1.4){\scriptsize 2};
    \node[] at (-0.5, 0.6){\scriptsize 1};
    \node[] at (0.5, 0.15){\scriptsize 3};
\end{tikzpicture}
& 
\begin{tabular}{c|c}
Cost & $d_1$ \\
\hline
$b_1 s_1$ & $0$ \\
$b_1 s_2$ & $0$ \\
$b_2 s_1$ & $0$ \\
$b_2 s_2$ & $0$
\end{tabular}
\end{tblr}
    \caption{A market where there is no without-tip equilibrium. As there is only one courier, one of the two stores are not bought in any feasible allocation. At the purchase price $0$, the buyer who does not buy will want to buy from a store that does not sell in the without-tip regime. A with-tip equilibrium is given by $p_{s_1}=1,p_{s_2}=0,\mathbf{t}=0,w_{b_1s_1}=4,w_{b_1s_2}=w_{b_2s_1}=w_{b_2s_2}=0$ and $d_1$ deliver $s_1$ to $b_1$.
    \label{fig:without_tip_not_exists}}
\end{figure}
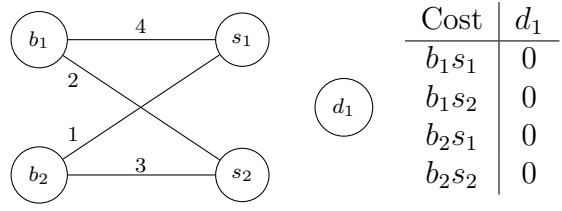

To analyze  equilibrium existence, we start from the courier side and look into a two-sided matching consisting of all orders and couriers. The following two lemmas show that when there are enough couriers, there is always a delivery compensation that incentivizes couriers to deliver any subset of orders. 
\cref{lem:exist_courier_plan_serve} also gives a tight upper bound on courier utility when delivering a subset of orders, which will be crucial in \cref{sec:general_markets} to test whether an allocation is supported in some equilibrium.

We need a little more notation to proceed. Let $G_D=(O,D)$ be a bipartite graph where each edge $(o,d)$ has cost $c_d(o)$. A \emph{courier allocation} $\mathbf{y}$ is a two-sided matching in $G_D$ where 
\begin{equation*}
    y_{od} =
    \begin{cases}
      1 & \text{if $d$ delivers the order $o$,}\\
      0 & \text{otherwise.}
    \end{cases}       
\end{equation*}

A feasible courier allocation $\mathbf{y}$ satisfies unit-capacity courier $\forall d, \sum_o y_{o,d}\leq 1$, unit-demand buyers and unit-supply store (i.e., $\forall o=(b,s),o'=(b',s') \text{ such that } y_{od}=y_{o'd}=1, \text{ it satisfies that } b\neq b, s\neq s'$). Denote $o_\mathbf{y}(d)$ as the order that courier $d$ delivers. If $d$ does not deliver, $o_\mathbf{y}(d)=\emptyset$. A \emph{courier plan} is a pair of delivery compensation and feasible courier allocation $(\mathbf{w},\mathbf{y})$ such that every courier delivers their most preferred order; i.e., $\forall d, o_{\mathbf{y}}(d)\in \BR_d(\mathbf{w})$.

Consider any subset of orders $\Omega\subset O$ where each buyer and store appears at most once.
A courier plan $(\mathbf{w},\mathbf{y})$ \emph{serves} $\Omega$ if all orders in $\Omega$ are delivered and only orders in $\Omega$ are delivered; i.e., $\forall o\in \Omega, \sum_{d}y_{od}=1 \text{ and } \forall o\notin \Omega, \sum_{d}y_{od}=0$.

\begin{restatable}{lemma}{existscourierPlan}
\label{lem:exist_courier_plan_serve}
    Consider the case where $|\Omega|< l$. Any courier plan $(\mathbf{w},\mathbf{y})$ that serves $\Omega$ satisfies (i) $\mathbf{y}$ is the minimum-cost matching that covers $\Omega$ in $G_D$; and (ii) courier utility is upper bounded by $$u_d(o_\mathbf{y}(d))\leq \bar{u}_d:=C_{\Omega}(G_D\backslash d)-C_{\Omega}(G_D)$$ where $C_{\Omega}(G_D)$ is the cost of the minimum cost matching that covers $\Omega$ in $G_D$, and $C_{\Omega}(G_D\backslash d)$ is the cost of the minimum-cost matching that covers $\Omega$ without courier $d$.
        
    There always exists a courier plan $(\bar{\mathbf{w}},\mathbf{y})$ that serves $\Omega$ and achieves $\bar{u}_d$, where $\bar{w}_{bs}=0$ for $(b,s)\notin \Omega$. 
\end{restatable}

\begin{restatable}{lemma}{existscourierPlanEqualSize}
\label{lem:exist_courier_plan_serve_equal_size}
    Consider the case where $|\Omega|=l$. Any courier plan $(\mathbf{w},\mathbf{y})$ that serves $\Omega$ satisfies that $\mathbf{y}$ is the minimum-cost matching that covers $\Omega$ in $G_D$. There always exists a courier plan $(\bar{\mathbf{w}},\mathbf{y})$ that serves $\Omega$ where $\bar{w}_{bs}=0$ for $(b,s)\notin \Omega$, and each courier has utility $$\bar{u}_d=H+ C_{\Omega_{-1}}(G_D\backslash d)-C_{\Omega}(G_D)\geq \max_{b,s}\{v_b(s)\},$$ where $H=\sum_{bs}v_b(s)+\sum_{bs}w_{bs}+\sum_{bsd}c_d(b,s)$, while $C_{\Omega_{-1}}(G_D\backslash d)$ is the cost of the minimum-cost matching that covers all but one order in $\Omega$, and $C_{\Omega}(G_D)$ is the cost of the minimum-cost matching that covers $\Omega$.
\end{restatable}

The proof of these two lemmas establishes a one-to-one relationship between courier plans that serve $\Omega$ and Walrasian equilibria in a two-sided market $M$ where one side is $\Omega$ and another side is $D$.
As a Walrasian equilibrium always exists in the market $M$, a courier plan that serves $\Omega$ always exists. By further appealing to the largest Walrasian price in $M$, we are able to write down the exact form of the maximum courier utility in a courier plan that serves $\Omega$. This maximum courier utility plays a crucial role in testing if an allocation is in some with-tip equilibrium in \cref{sec:general_markets}.

\begin{theorem}\label{thm:without_tip_eq_existence}
    In a market where the number of couriers is weakly larger than the minimum number of buyers and stores $l\geq \min\{m,n\}$, a without-tip equilibrium always exists.
\end{theorem}
\begin{proof}
    Consider a Walrasian equilibrium $(\mathbf{p},\mathbf{z})$ in the two-sided market $(B,S)$. $\mathbf{z}$ defines a subset of orders $\Omega =\{(b,s) | z_{bs}=1\}$ satisfying $|\Omega|\leq l$. By \cref{lem:exist_courier_plan_serve} and \cref{lem:exist_courier_plan_serve_equal_size}, there is always a courier plan $(\mathbf{w},\mathbf{y})$ that serves $\Omega$ and satisfies $\forall (b,s)\notin \Omega, \; w_{bs}=0$. Define a feasible allocation $\mathbf{x}$, where for an order $(b,s)=o$, let $x_{bsd}= z_{bs}y_{od}$. $\mathbf{x}$ does not change what buyers buy in $(\mathbf{p},\mathbf{z})$ and does not change what couriers deliver in $(\mathbf{w},\mathbf{y})$. Then $(\mathbf{p},\mathbf{w},\mathbf{x})$ is a without-tip equilibrium. This is because the Walrasian equilirbium guarantees buyers incentives, the courier plan guarantees couriers incentives, and unsold store prices and undelivered orders delivery compensations are zero. 
\end{proof}

\begin{theorem}\label{thm:with_tip_eq_existence}
    There always exists a with-tip equilibrium.
\end{theorem}
\begin{proof}
    \cref{lem:pareto_dominant} shows that any allocation in a without-tip equilibrium is in a with-tip equilibrium of zero tips. When $l\geq \min\{m,n\}$, \cref{thm:without_tip_eq_existence} says without-tip equilibria exist. So a with-tip equilibria also exists. When $l<\min\{m,n\}$, consider a subset of orders $\Omega$ where each buyer and store appears at most once and $|\Omega|=l$. \cref{lem:exist_courier_plan_serve_equal_size} says there exists a courier plan $(\bar{\mathbf{w}},\mathbf{y})$ that serves $\Omega$ where $\bar{w}_{bs}=0$ for $(b,s)\notin \Omega$, and satisfies all courier utility being large $u_d=\bar{u}_d>\max_{bs}\{v_b(s)\}$. Define an allocation $\mathbf{x}$
    \begin{equation*}
        x_{bsd}=
        \begin{cases}
            1 & \text{ if } (b,s)=o\in\Omega \text{ and } y_{od}=1\\
            0 & \text{ otherwise}
        \end{cases}
    \end{equation*} 
    Define a purchase price
    \begin{equation*}
        p_{s}=
        \begin{cases}
            v_b(s) & \text{ if } \exists (b,s)\in\Omega\\
            0 & \text{ otherwise}
        \end{cases}
    \end{equation*} 
    Set tips to be zero $\mathbf{t}=0$. Then $(\mathbf{p},\bar{\mathbf{w}},\mathbf{t},\mathbf{x})$ is a with-tip equilibrium. To see this, first it satisfies that unsold store has purchase price zero, and undelivered orders have compensation and tips zero. Couriers incentives are satisfied from the definition of courier plan. For any buyer $b$ that buys from store $s_\mathbf{x}(b)$, the minimum tip for store $s'\neq s_\mathbf{x}(b)$ is larger than its valuation, i.e., $\underline{t}_{bs'}>v_b(s')$. This is because to have any courier $d$ delivers from $s'$ to $b$, the tip must satisfy $t+p_{bs'}=t\geq \bar{u}_d>\max_{bs}\{v_b(s)\}$. However, buyers $b$ only needs to pay zero tips for $s_\mathbf{x}(b)$: $\underline{t}_{bs_\mathbf{x}(b)}=t_{bs_\mathbf{x}(b)}=0$. Buyers incentives are also satisfied.
\end{proof}

Although equilibria always exists for with-tip equilibrium, they can be inefficient. We will discuss the welfare aspect of equilibrium in  \cref{sec:general_markets} and \cref{sec:market_structures}.

\subsection{The Role of Tips}\label{sec:role_tips}
We now further illustrate the role of tips in pricing. Recall  \cref{example:without_tip_bad}, where the only without-tip equilibrium allocation is also supported in a with-tip equilibrium. We show this is not a coincidence.
\begin{restatable}{proposition}{paretoDominant}
\label{lem:pareto_dominant}
    For a market, if $(\mathbf{p},\mathbf{w},\mathbf{x})$ is a without-tip equilibrium, then $(\mathbf{p},\mathbf{w},\mathbf{t},\mathbf{x})$ is a with-tip equilibrium where $\mathbf{t}=0$. 
\end{restatable}
The proof is relatively straightforward. At zero tips, the set of orders that maximize courier utility remains the same. A buyer $b$ who buys from store $s_\mathbf{x}(b)$ does not need to pay tips to get delivery from $s_\mathbf{x}(b)$, i.e., $\underline{t}_{bs_\mathbf{x}(b)}=0$. But $b$ does need to pay tips weakly larger than zero to get delivery for some other stores $s'\neq s_\mathbf{x}(b)$. So $s_\mathbf{x}(b)$ is still in the set of stores that maximize buyer incentives.

This proposition implies that with-tip equilibria weakly Pareto dominate without-tip equilibria. From a welfare perspective, the benefit of including tips in pricing is made even more apparent when we do not require an equilibrium to clear the market; i.e., unsold stores can have positive price and undelivered orders can have positive compensations and tips. In \cref{app:tips_eq_not_clear_market}, we prove that when equilibria do not need to clear the market, the optimal with-tip equilibrium welfare is always lower bounded above zero while the optimal without-tip equilibrium welfare can be zero.

While it may seem intuitive that the added flexibility of tips leads to higher welfare, we will see in \cref{lem:zero_tip_equivalence} that the welfare gain does not come from the actual (on-path) payment of tips, but rather from the off-path need of buyers to offer tips to secure delivery, pinpointing the role of tips as preventing deviations off path rather than for on-path transfers.
\begin{restatable}{proposition}{ZeroTipEquivalence}
\label{lem:zero_tip_equivalence}
    If there exists a with-tip equilibrium $(\mathbf{p},\mathbf{w},\mathbf{t},\mathbf{x})$, then there exists a with-tip equilibrium $(\mathbf{p}',\mathbf{w}',\mathbf{t}',\mathbf{x})$ with zero tips $\mathbf{t}'=0$.
\end{restatable}
To prove this, we construct $\mathbf{p}',\mathbf{w}'$ in a way that directly incorporates tips $\mathbf{t}$ into the purchase price and delivery compensation. That is, for $x_{bsd}=1$, set $p'_{s}=p_s+t_{bs}$ and $w'_{bs}=w_{bs}+t_{bs}$. The crux of this proof is to show that the minimum tip for a buyer $b$ to get delivery from a store $s\neq s_\mathbf{x}(b)$ is the same between $(\mathbf{p},\mathbf{w},\mathbf{t},\mathbf{x})$ and $(\mathbf{p}',\mathbf{w}',\mathbf{t}',\mathbf{x})$. With the same minimum tip, we can prove that $s_\mathbf{x}(b)$ is still in the set of stores that maximize buyer utility.
The proposition will play a crucial role in \cref{sec:general_markets} to determining whether an allocation is supported in some with-tip equilibrium.

\cref{lem:pareto_dominant,lem:zero_tip_equivalence} together offer an interpretation of the role of tips on delivery platforms. In the with-tip equilibrium, the platform can set prices and compensations that require buyers to pay zero tips in equilibrium, but very high tip should buyers deviate. In this sense, tips serves as a tool for price discrimination against buyers who deviate from equilibrium. It is interesting to note that ride-sharing platforms have been reported to engage in price discrimination even regardless of tips \citep{pandey2021disparate}, charging different prices for identical trips for different riders\footnote{See online reports from \url{https://www.reddit.com/r/uber/comments/1jpe7ks/charging_different_prices_for_the_exact_same_ride/,https://www.reddit.com/r/uber/comments/1629oz9/price_discrimination/, https://www.theguardian.com/commentisfree/2018/apr/13/uber-lyft-prices-personalized-data, https://iddp.gwu.edu/researchers-find-racial-discrimination-dynamic-pricing-algorithms-used-uber-lyft-and-others}}. One might speculate that for this reason, the ability to use tips for price discrimination would be less valuable for the platform in such settings, and indeed tips are only specified after delivery for such platforms.

\section{Markets without any Structural Constraints}\label{sec:general_markets}

In the sequel, we focus on the welfare properties of the with-tip equilibrium. For simplicity, we use ``equilibrium'' to refer to a with-tip equilibrium and ``without-tip equilibrium'' to denote the alternative. 

We first introduce a market where all equilibria are inefficient. We  prove that computing the optimal equilibrium welfare, as well as the optimal welfare over all feasible allocations, is NP-hard. Given this computational challenge, we turn to characterizing all equilibrium allocations in \cref{sec:char_eq_alloc} and show that for any allocation, we can check whether there exists an equilibrium that contains this allocation in polynomial time. Most proofs are presented in \cref{app:general_markets}.

\begin{example}
    The market in Figure~\ref{fig:market_clearing} is an example where all equilibria have low welfare compared with the optimal welfare. Let $\kappa>2$ be a constant.  Both buyers value both stores at $1$, and the two couriers have costs $c_{d_1}(b_1s_1)=0,c_{d_1}(b_1s_2)=c_{d_1}(b_2s_1)=c_{d_2}(b_1s_1)=c_{d_2}(b_2s_2)=\kappa,c_{d_1}(b_2s_2)=c_{d_2}(b_2s_1)=0.5,c_{d_2}(b_1s_2)=0.49$. The optimal welfare in this market is $1$, realized by $b_1$ buys from $s_1$ delivered by $d_1$. We show the optimal equilibrium welfare is $2-\kappa$ and can be very inefficient when $\kappa$ is large.

    Consider allocations where both buyers buy. These allocations have a maximum welfare of $2-\kappa<0$. One such equilibrium is specified by $x_{b_2s_2d_2}=x_{b_1s_1d_1}=1,p_{s_1}=p_{s_2}=1,w_{b_1s_2}=w_{b_2s_1}=0,w_{b_1s_1}=w_{b_2s_2}=\kappa, \mathbf{t}=0$.

    We now show there is no equilibrium where only one buyer buys. If such an equilibrium exists, the store not bought from has zero purchase price. And the courier with the lower cost between $d_1$ and $d_2$ delivers. Otherwise, the idle courier with a lower cost will want to deliver as well.
    \begin{itemize}
        \item If only $b_1$ buys from $s_1$. $d_1$ delivers, $d_2$ has zero utility. Buyer $b_1$ can buy from $s_2$ with the minimum tip that makes $d_2$ indifferent to deliver $\underline{t}_{b_1s_2}=0.49$, resulting in a buyer utility of $0.51$. To make $b_1$ buy from $s_1$, $p_1\leq 0.49$. But then $b_2$ can buy $s_1$ with tip $0.5$, resulting in a positive utility $1-0.5-0.49>0$.
        \item If only $b_1$ buys from $s_2$. $d_2$ delivers, $d_1$ has zero utility. $b_1$ can buy $s_1$ with zero tip and utility $1$, so $p_2=0$. But then $b_2$ can tip $d_1$ with $\underline{t}_{b_2s_2}=0.5$ to buy from $s_2$ with positive utility $1-0.5>0$.
    \end{itemize}
    With the same logic, one can show the other two allocations ``Only $b_2$ buys from $s_1$'' and ``Only $b_2$ buys from $s_2$'' are not in any equilibrium.
\end{example}
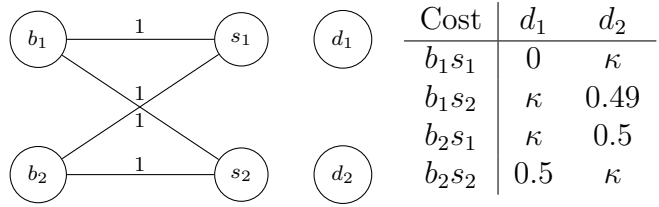
\begin{figure}[t]
    \centering
    \begin{tblr}{}
    \centering
\begin{tikzpicture}[baseline=(current bounding box.center),scale=0.9]
    \node[draw, shape=circle, minimum size=0.1cm] (1) at (-1, 2) {\scriptsize $b_1$};
    \node[draw, shape=circle, minimum size=0.1cm] (2) at (-1, 0) {\scriptsize $b_2$};
    \node[draw, shape=circle, minimum size=0.1cm] (3) at (2, 2) {\scriptsize $s_1$};
    \node[draw, shape=circle, minimum size=0.1cm] (4) at (2, 0) {\scriptsize $s_2$};
    \node[draw, shape=circle, minimum size=0.1cm] (5) at (3.5, 2) {\scriptsize $d_1$};
    \node[draw, shape=circle, minimum size=0.1cm] (6) at (3.5, 0) {\scriptsize $d_2$};
    \draw (1) to (3);
    \draw (1) to (4);
    \draw (2) to (3);
    \draw (2) to (4);
    \node[] at (0.5, 2.2){\scriptsize 1};
    \node[] at (0.5, 1.2){\scriptsize 1};
    \node[] at (0.5, 0.8){\scriptsize 1};
    \node[] at (0.5, 0.15){\scriptsize 1};
\end{tikzpicture}
& 
\begin{tabular}{c|cc}
Cost & $d_1$  & $d_2$\\
\hline
$b_1 s_1$ & $0$  & $\kappa$\\
$b_1 s_2$ & $\kappa$  & $0.49$\\
$b_2 s_1$ & $\kappa$  & $0.5$\\
$b_2 s_2$ & $0.5$  & $\kappa$
\end{tabular}
\end{tblr}
    \caption{A market where all equilibria are inefficient. Both buyers value both stores at $1$. The two courier have different cost of delivery.
    \label{fig:market_clearing}}
\end{figure}

We show through a simple reduction from 3-dimensional matching (3DM), that finding the optimal equilibrium welfare is NP-hard.
\begin{theorem}\label{lem:OPT_eq_hard}
    It is NP-hard to determine if the optimal equilibrium welfare is equal to $l$ in a market where there is equal number of buyers, stores and couriers $m=n=l$, and all buyers' valuations are 1, and all couriers' costs are in $\{0,1\}$.
\end{theorem}
\begin{proof}
    We reduce from the 3DM, one of the first 21 NP-complete problems proved \citep{karp1975computational}. 3DM asks the following question: Given a 3-regular hypergraph $G=(B,S,D)$ with hyperedges $T\subset B \times S \times D$ where $|B|=|S|=|D|=l$, is there a perfect matching? 
    
    We construct a following market $M$ with $m$ buyers, $n$ stores and $l$ couriers where $m=n=l$. Buyers' valuations are always 1. Couriers' costs are set as 
    \begin{equation*}
    c_d(b,s) =
    \begin{cases}
      0 & \text{if $ (b,s,d)\in T$,}\\
      1 & \text{otherwise.}
    \end{cases}       
    \end{equation*}
  
    If $G$ has a perfect matching of size $l$, this perfect matching defines a subset of orders $\Omega=\{(b,s) |\; \exists d \text{ such that } (b,s,d) \text{ is in the perfect matching}\}$. It holds that $|\Omega|=l$. By \cref{lem:exist_courier_plan_serve_equal_size}, there exists a courier plan $(\bar{\mathbf{w}},\mathbf{y})$ that serves $\Omega$ where all orders $(b,s)\notin \Omega$ has delivery compensation $\bar{w}_{bs}=0$. We  define a purchase price $\mathbf{p}=1$, a tip $\mathbf{t}=0$, an allocation $\x$ where $x_{bsd}=1$ iff and only if $(b,s)=o\in\Omega \text{ and } y_{od}=1$.
    
    Then $(\mathbf{p},\bar{\mathbf{w}},\mathbf{t},\mathbf{x})$ is an equilibrium. This is because all buyers buy with minimum tip 0 and utility 0. If a buyer $b$ buy from a store other than $s_\mathbf{x}(b)$, she needs to pay nonzero tips and a purchase price of 1, resulting in a deviation utility weakly less than 0. The construction of courier plan guarantees courier incentives. By our construction, this equilibrium has welfare $n$.

   For the other direction, If $G$ does not have a perfect matching of size $n$, then by the construction of the market $M$, there does not exists any allocation $\mathbf{x}$ with welfare $n$. So the optimal equilibrium welfare is smaller than $n$. We have shown that $G$ has a perfect matching if and only if there exists an equilibrium with welfare $l$. This completes the proof.
\end{proof}
The computational complexity is not only limited to finding the optimal equilibrium welfare. With the same construction, it can be shown that finding the optimal welfare is also hard. 
\begin{restatable}{theorem}{OPTHard}
    It is NP-hard to determine if the optimal welfare is equal to $l$ in a market where there is equal number of buyers, stores and couriers $m=n=l$, and all buyers' valuations are 1, and all couriers' costs are in $\{0,1\}$.
\end{restatable}

Seeing the computational complexity, one natural direction is to extend the approximation results in the maximum-weight version of 3DM to our setting, and develop a polynomial-time algorithm that finds equilibria with good welfare. However, most approximations results for maximum-weight 3DM rely on local-search algorithms that swap subsets of individual matches. These local-search algorithms do not consider buyers and couriers' incentives, and the swaps
break the equilibrium requirements, making them unsuitable for our economic setting.

\subsection{Characterizing Equilibrium allocations}\label{sec:char_eq_alloc}

Given this computational challenge, we turn to characterizing the set of equilibrium allocations. Given a feasible allocation $\mathbf{x}$, we ask whether there is an efficient way to determine if there is an equilibrium $(\mathbf{p},\mathbf{w},\mathbf{t},\mathbf{x})$ for some $\mathbf{p},\mathbf{w},\mathbf{t}$. We answer this question affirmatively.

A feasible allocation $\mathbf{x}$ induces a subset of orders $\Omega_\mathbf{x}=\{(b,s)|\sum_d x_{bsd}=1\}$ to be delivered by couriers. Since this allocation is feasible, $|\Omega_\mathbf{x}|\leq l$. \cref{lem:zero_tip_equivalence} says it is without loss to set $\bt=0$. So given $\x$, we only need to check for the existence of $\p,\w$. Furthermore by the equilibrium requirements, $w_{bs}=0$ for $(b,s)\notin \Omega_\x$. 

\cref{lem:exist_courier_plan_serve} and \cref{lem:exist_courier_plan_serve_equal_size} show the existence of a courier plan $(\bar{\mathbf{w}},\mathbf{y})$ that serves $\Omega_\x$ and where couriers achieve utility $\bar{u}_d$. As $\bt=0$, when $(b,s)\in\Omega_\x$, the courier plan serving $\Omega_\x$ already guarantees that some couriers are willing to deliver $(b,s)$ so $\ut_{bs}=0$. When $(b,s)\notin\Omega_\x$, to have courier $d$ deliver for store $s$, the tip must be large enough to match couriers current utility in the courier plan $$\ut_{bs}+w_{bs}-c_d(b,s)=\ut_{bs}+0-c_d(b,s) \geq \bar{u}_d.$$

The minimum tip buyer $b$ needs to offer for some courier to deliver from a store $s$ is
\begin{equation}\label{eq:highest_minimum_tip}
    \ut_{bs} = \begin{cases}
        0 & \text{if } (b,s)\in\Omega_x, \\
        \min_d c_d(b,s)+\bar{u}_d & \text{otherwise.}
    \end{cases}     
\end{equation}

When $|\Omega_\x|<l$, $\bar{u}_d$ is the highest possible utility of courier $d$ in any courier plan that serves $\Omega_\x$. Then for an order $(b,s)\notin\Omega_\x$, $\ut_{bs}$ is the highest possible minimum tip required for some couriers to serve this order. When $|\Omega_\x|=l$, $\bar{u}_d >\max_{bs}v_b(s)$ so $\ut_{bs}$ is higher than any buyer valuation. The role of this high tips is to remove a buyer $b$'s incentives to break equilibrium and purchase from another store $s'\neq s_\x(b)$. The following lemma formalizes this intuition.

\begin{restatable}{lemma}{highestcourierPriceEq}
\label{lem:highest_courier_price_eq}
    If a feasible allocation $\mathbf{x}$ is in some equilibrium $(\p,\w,\bt,\x)$ where $\bt=0$, then $(\p,\bar{\w},\bt,\x)$ is also an equilibrium. 
\end{restatable}
The proof first shows that $\bar{\w}$ satisfies courier incentives. It then reasons that since a buyer $b$ does not offer a lower tip to buy from any other store than $s_\x(b)$ in $(\p,\w,\bt,\x)$, she does not offer a higher tip to buy from other stores in $(\p,\bar{\w},\bt,\x)$.

We have shown through \cref{lem:zero_tip_equivalence} and \cref{lem:highest_courier_price_eq} that given an allocation $\mathbf{x}$, it is without loss to check for an equilibrium where $\bt=0$ and delivery compensation equals to $\bar{\w}$. We are now ready to check if $\x$ is in some equilibrium.

Define a buyer allocation $\mathbf{z}$ that allocates stores to buyers to be $z_{bs}=\sum_d x_{bsd}$. Let $G_\x=(B,S)$ be the bipartite graph where one side are buyers, and the other side are stores. Each edge $(b,s)$ in $G_\x$ has weight 
\begin{equation*}
    v_{bs}^\x = \begin{cases}
        v_b(s) - \ut_{bs} & \text{ if } (b,s)\in\Omega_\x,\\
        v_b(s) & \text{otherwise.}
    \end{cases} 
\end{equation*}

Here, we  use $v^\x_{bs}$ to denote the edge weight in $G_\x$, because this edge weight corresponds to a buyer $b$'s realized value when buying from store $s$, after adjusting for the tips $\ut_{bs}$ she has to pay to get the delivery. 

\begin{restatable}{lemma}{testAlloc}
\label{lem:test_alloc}
    A feasible allocation $\mathbf{x}$ is in some equilibrium if and only if $\mathbf{z}$ is a maximum weight matching in $G_\x$. 
\end{restatable}
The proof for this lemma shows that $\mathbf{x}$ is an equilibrium if and only if $\mathbf{z}$ is a Walrasian equilibrium allocation in the two-sided market $G_x$, where buyer valuations for stores are adjusted by tips.

\cref{lem:test_alloc} offers a way to test if $\x$ is an equilibrium allocation from a maximum weight matching perspective. We rely on this lemma in \cref{sec:market_structures} and \cref{sec:market_structures_buyers} to test if the welfare-optimal allocation is in an equilibrium. An immediate corollary is that any feasible allocation $\x$ where all couriers make a delivery is supported in some equilibrium. This is because the platform can subsidize couriers, so that the minimum tip $\ut_{bs}$ is much higher than buyer $b$'s valuations for $s\neq s_\x(b)$.
\begin{restatable}{corollary}{allcouriersDeliver}
\label{cor:all_courier_deliver}
    If a feasible allocation $\x$ satisfies that $|\Omega_\x|=l$, then $\x$ is in some equilibrium. 
\end{restatable}

We summarize our results in this section in the following theorem. As the computation for $\bar{u}_d$ and $\ut_{bs}$ for every buyer, store, and courier takes polynomial time, we have:
\begin{restatable}{theorem}{PolyCheck}
\label{thm:poly_check}
    For any allocation $\mathbf{x}$, determining if there exists an equilibrium that includes $\mathbf{x}$ takes polynomial time.
\end{restatable}

\section{Markets with Structures}
While the previous section demonstrates that all equilibria can be inefficient and it is computationally intractable to find the optimal equilibrium welfare, in this section we identify structures on couriers' costs or buyers' valuations that overcome both impossibilities. We continue to use ``equilibrium'' to refer to with-tip equilibrium and ``without-tip'' equilibrium to denote the alternative.

\subsection{Structures on courier Costs}\label{sec:market_structures}
The structure on courier costs identifies components of cost. A courier's cost $c_d(b,s)$ can consist of three parts: (i) a buyer--store part $c(b,s)$, (ii) a courier--store part $c_d(s)$, and (iii) a courier--buyer part $c_d(b)$. These reflect distance-based costs that account for couriers' deliveries to buyers, trips to stores, and return-home trips. 
The courier--store cost can also reflect a courier's familiarity with a store, while the courier--buyer cost can account for familiarity with the buyer's neighborhood.

We first show that when for all couriers $c_d(b,s)=c(b,s)+c_d(s)$ or for all couriers $c_d(b,s)=c(b,s)+c_d(b)$, i.e., when one of the two possible structural assumptions holds, the optimal welfare across all allocations can be found in polynomial time.

\begin{restatable}{theorem}{divisiblecourierCostOpt}
    \label{thm:divisible_courier_cost_opt}
    The optimal welfare can be computed in polynomial time when the couriers' costs are decomposable into a buyer--store and a courier--store part $c_d(b,s)=c(b,s)+c_d(s)$, or a buyer--store and a courier--buyer part $c_d(b,s)=c(b,s)+c_d(b)$.
\end{restatable}
\begin{proof}
    We reduce the problem of finding the optimal welfare to a standard \emph{minimum cost flow} problem, solvable in polynomial time using existing algorithms. We demonstrate this for $c_d(b,s)=c(b,s)+c_d(s)$, while the similar case $c_d(b,s)=c(b,s)+c_d(b)$ is omitted. 
    
    Given a market, construct a flow network. Figure~\ref{fig:OPTDivisibleCost} illustrates such a network for a market with two buyers, two stores and three couriers. There is a source node $S_o$, a sink node $S_i$, a node for each buyer $b\in B$, order $o\in B\times S$, store $s\in S$, courier $d\in D$, and an additional dummy node for each store. There exists a directed edge from the source to each buyer node, from each buyer to each order that involves the buyer, from each order to the store that is involved in the order, from each store to the dummy node of the same store, from each dummy store to each courier, and from each courier to the sink. All edges have capacity one. An edge between a buyer $b$ and an order $(b,s)$ has cost $-v_b(s)$, an edge between an order $(b,s)$ and a store $s$ has cost $c(b,s)$, an edge between a dummy store $s'$ and a courier have $c_d(s)$ as cost. The net supply or demand are zero for all nodes except the source and the sink node.

    An allocation is feasible if and only if there is a corresponding integer flow. Given a flow, define an allocation $\x$ in the following way. Each unit of flow goes through a path $So\rightarrow b\rightarrow (b,s)\rightarrow s\rightarrow s'\rightarrow d\rightarrow Si$ for some $b,s,d$. Set $x_{bsd}=1$. $\x$ is feasible because all edge capacity is one and the flow is integral. Given an allocation, for each triplet where $x_{bsd}=1$, push one unit of flow through the path $So\rightarrow b\rightarrow (b,s)\rightarrow s\rightarrow s'\rightarrow d\rightarrow Si$. This defines an integer flow.
    
    To solve for the optimal welfare, set net supply of the source node equal to the net demand of the sink node to be $f=1,2,...,\min\{m,n,l\}$. Solve the minimum cost flow problem for each $f$ and take the flow with the minimum cost over all $f$. By the integrality theorem, any minimum cost network flow problem instance with integral edge capacities, costs, and supply has an optimal integer solution. As the cost of a flow equals to the negation of the welfare of an allocation, finding the minimum cost flow among all supply $f=1,2,...,\min\{m,n,l\}$ is finding the allocation with the optimal welfare.
\end{proof}

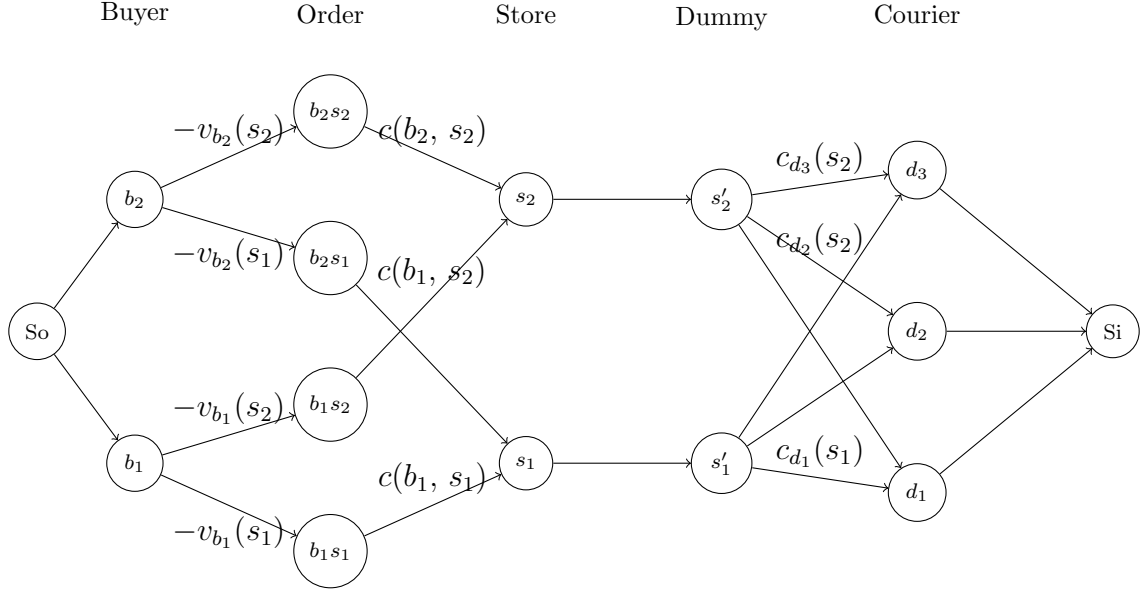
\begin{figure}[t] 
    \centering
    \begin{tikzpicture}[scale=1.3]
        \node[draw, shape=circle] (Source) at (0, 2.25) {\scriptsize So};
        \node[draw, shape=circle] (Sink) at (11, 2.25) {\scriptsize Si};
        \foreach \i/\label in {0.9/$b_1$, 3.6/$b_2$}
            \node[draw, shape=circle,minimum size=0.1cm] (\label) at (1, \i){\scriptsize\label};
        \foreach \i/\label in {0/$b_1 s_1$, 1.5/$b_1 s_2$, 3/$b_2 s_1$, 4.5/$b_2 s_2$}
            \node[draw, shape=circle,minimum size=0.1cm] (\label) at (3, \i){\scriptsize\label};
        \foreach \i/\label in {0.9/$s_1$, 3.6/$s_2$}
            \node[draw, shape=circle,minimum size=0.1cm] (\label) at (5, \i){\scriptsize\label};
        \foreach \i/\label in {0.9/$s'_1$, 3.6/$s'_2$}
            \node[draw, shape=circle,minimum size=0.1cm] (\label) at (7, \i){\scriptsize\label};
        \foreach \i/\label in {0.6/$d_1$, 2.25/$d_2$, 3.9/$d_3$}
            \node[draw, shape=circle,minimum size=0.1cm] (\label) at (9, \i){\scriptsize\label};

        \foreach \x/\y in {Source/$b_1$, Source/$b_2$}
            \draw[->] (\x) -- node[midway, right]{} (\y);
        \foreach \x/\y/\w in {$b_1$/$b_1 s_2$/$-v_{b_1}(s_2)$, $b_2$/$b_2 s_2$/$-v_{b_2}(s_2)$}
            \draw[->] (\x) -- node[above]{\w} (\y);
        \foreach \x/\y/\w in {$b_1$/$b_1 s_1$/$-v_{b_1}(s_1)$, $b_2$/$b_2 s_1$/$-v_{b_2}(s_1)$}
            \draw[->] (\x) -- node[below]{\w} (\y);

        \foreach \x/\y/\w in {$b_2 s_2$/$s_2$/$c(b_2\text{, }s_2)$, $b_2 s_1$/$s_1$/,$b_1 s_2$/$s_2$/$c(b_1\text{, }s_2)$, $b_1 s_1$/$s_1$/$c(b_1\text{, }s_1)$}
            \draw[->] (\x) -- node[above]{\w} (\y);
        \foreach \x/\y in {$s_2$/$s'_2$, $s_1$/$s'_1$}
            \draw[->] (\x) -- node[midway]{} (\y);
        \foreach \x/\y/\w in {$s'_2$/$d_3$/$c_{d_3}(s_2)$, $s'_2$/$d_2$/$c_{d_2}(s_2)$,$s'_2$/$d_1$/, $s'_1$/$d_3$/,$s'_1$/$d_2$/, $s'_1$/$d_1$/$c_{d_1}(s_1)$}
            \draw[->] (\x) -- node[above]{\w} (\y);
        \foreach \x in {$d_1$,$d_2$,$d_3$}
            \draw[->] (\x) -- node[]{}(Sink);
            
        \node[] at (1, 5.5)  {\footnotesize Buyer};
        \node[] at (3, 5.5)  {\footnotesize Order};
        \node[] at (5, 5.5)  {\footnotesize Store};
        \node[] at (7, 5.45)  {\footnotesize Dummy};
        \node[] at (9, 5.5)  {\footnotesize Courier};

    \end{tikzpicture}
    \caption{An example min-cost flow network for a market where $c_d(b,s)=c(b,s)+c_d(s)$. The market has two buyers, two stores, and three couriers. There is a dummy vertex $s'$ for each store $s$ to enforce the store's unit-capacity constraint. Each edge has capacity one. An edge between a buyer $b$ and an order $(b,s)$ has cost $-v_b(s)$. An edge between an order $(b,s)$ and a store $s$ have cost $c(b,s)$. An edge between a dummy vertex of a store $s$ and a courier $d$ has cost $c_d(s)$.
    \label{fig:OPTDivisibleCost}}
    \end{figure}

In the next theorem we show when $c_d(b,s)=c(b,s)+c_d(s)$, or $c_d(b,s)=c(b,s)+c_d(b)$, there always exists an equilibrium that achieves the optimal welfare. This is not true for without-tip equilibria. Recall the market in \cref{fig:without_tip_not_exists} does not have a without-tip equilibrium, despite all courier costs being zero. And even when a without-tip equilibrium does exist--- as in \cref{example:without_tip_bad}--- it can still yield very low welfare, despite the structured courier costs. Note that courier costs in \cref{example:without_tip_bad}  can be decomposes either as $c(b_1,s_1)=0, c(b_2,s_1)=11, c_{d_1}(s_1)=0, c_{d_2}(s_1)=1$, or as $c(b_1,s_1)=0, c(b_2,s_1)=11, c_{d_1}(b_1)=c_{d_1}(b_2)=0, c_{d_2}(b_1)=c_{d_2}(b_2)=1$.

\begin{restatable}{theorem}{DivisiblecourierCost}
\label{thm:divisible_courier_cost}
    When the courier costs consist of a buyer--store and a courier--store part $c_d(b,s)=c(b,s)+c_d(s)$, or a buyer--store and a courier--buyer part $c_d(b,s)=c(b,s)+c_d(b)$, there always exists a with-tip equilibrium that achieves the optimal welfare. However, there exists a market with such courier costs, where every without-tip equilibrium has arbitrarily low welfare.
\end{restatable}
\begin{proof}
    Let $\x$ be a feasible allocation that achieves the optimal welfare $W(\x)=\mathit{OPT}$. If $|\Omega_\x|=l$, \cref{cor:all_courier_deliver} shows $\x$ is in some equilibrium. We focus on the case where $|\Omega_\x|<l$. We present the proofs here for $c_d(b,s)=c(b,s)+c_d(b)$ and all stores are bought from $|\Omega_x|=n$. We defer the case for $|\Omega_x|<n$ and $c_d(b,s)=c(b,s)+c_d(s)$ to the \cref{app:market_structures}.
    
    Let $\z$ be the buyer allocation induced by $
    \x$ in $G_\x$, and $\z'$ be any other feasible buyer allocation in $G_\x$. We will show that $\z$ is of weakly larger weight than $\z'$. Then \cref{lem:test_alloc} shows $\x$ is in an equilibrium. In $G_\x$, the two matching $\z$ and $\z'$ define some alternating paths and cycles $$\pi=((b_0),s_1,b_1,s_2,b_2,\ldots, b_{t-1},s_t, b_t),$$ where for every $q\in \{1,\ldots, t\}$, $$z_{b_q,s_q}=1\mbox{ and }z_{b_{q-1},s_q}=0 \mbox{ and }z'_{b_q,s_q}=0\mbox{ and } z'_{b_{q-1},s_q}=1.$$ 
    
    An alternating path can either start from 1) $b_0$, a buyer allocated by $\z'$ but not by $\z$; 2) $s_1$, a store allocated by $\z$ but not by $\z'$. In an alternating path buyer $b_t$ is allocated by $\z$ but not by $\z'$. An alternating cycle is defined by $b_t=b_0$, and all buyers and stores on the cycle is allocated by both $\z$ and $\z'$. We use $\pi$ to denote both alternating paths and cycles. As $|\Omega_\x|=n$, all stores in $\pi$ are allocated in $\z$. Each alternating path and cycle captures buyers comparing the store they buy from in $\x$, against stores they can buy from by paying tips.
    Figure~\ref{fig:alternative_path_cycle} shows some examples of the alternating paths and cycles.

    \begin{figure}[t]
    \centering
    \begin{subfigure}{.3\textwidth}
    \centering
    \begin{tikzpicture}[scale=0.7]
        \foreach \i/\label in {1/$b_1$, 2.5/$b_2$, 4/$b_3$}
            \node[draw, shape=rectangle,scale=0.8] (\label) at (\i, 2) {\tiny\label};
            
        \foreach \i/\label in {1/$s_1$, 2.5/$s_2$, 4/$s_3$}
            \node[draw, shape=circle, scale=0.8] (\label) at (\i, 0) {\tiny\label};

        \foreach \x/\y in {$b_1$/$s_2$, $b_2$/$s_3$}
            \draw[line width=0.5pt, dashed] (\x) -- (\y);

        \foreach \x/\y in {$b_1$/$s_1$, $b_2$/$s_2$, $b_3$/$s_3$}
            \draw[line width=0.5pt, black, solid] (\x) -- (\y);    
    \end{tikzpicture}
    \end{subfigure}%
    \begin{subfigure}{.3\textwidth}
    \centering
    \begin{tikzpicture}[scale=0.7]
        \foreach \i/\label in {-0.5/$b_0$,1/$b_1$, 2.5/$b_2$, 4/$b_3$}
            \node[draw, shape=rectangle,scale=0.8] (\label) at (\i, 2) {\tiny\label};
            
        \foreach \i/\label in {1/$s_1$, 2.5/$s_2$, 4/$s_3$}
            \node[draw, shape=circle, scale=0.8] (\label) at (\i, 0) {\tiny\label};

        \foreach \x/\y in {$b_0$/$s_1$,$b_1$/$s_2$, $b_2$/$s_3$}
            \draw[line width=0.5pt, dashed] (\x) -- (\y);

        \foreach \x/\y in {$b_1$/$s_1$, $b_2$/$s_2$, $b_3$/$s_3$}
            \draw[line width=0.5pt, black, solid] (\x) -- (\y);    
    \end{tikzpicture}
    \end{subfigure}
    \begin{subfigure}{.3\textwidth}
    \centering
    \begin{tikzpicture}[scale=0.7]
        \foreach \i/\label in {1/$b_1$, 2.5/$b_2$, 4/$b_3$}
            \node[draw, shape=rectangle,scale=0.8] (\label) at (\i, 2) {\tiny\label};
            
        \foreach \i/\label in {1/$s_1$, 2.5/$s_2$, 4/$s_3$}
            \node[draw, shape=circle, scale=0.8] (\label) at (\i, 0) {\tiny\label};

        \foreach \x/\y in {$b_1$/$s_2$, $b_2$/$s_3$, $b_3$/$s_1$}
            \draw[line width=0.5pt, dashed] (\x) -- (\y);

        \foreach \x/\y in {$b_1$/$s_1$, $b_2$/$s_2$, $b_3$/$s_3$}
            \draw[line width=0.5pt, black, solid] (\x) -- (\y);    
    \end{tikzpicture}
    \end{subfigure}
    \caption{In $G_x$, $\z$ and $\z'$ define some alternating paths and cycles that do no intersect. Solid edges means $z_{bs}=1$, dotted edges means $z'_{bs}=1$. From left to right: alternating path with no unmatched buyer in $\z$, alternating path with $b_0$ unmatched by $\z$, alternating cycle.\label{fig:alternative_path_cycle}}
    \end{figure}
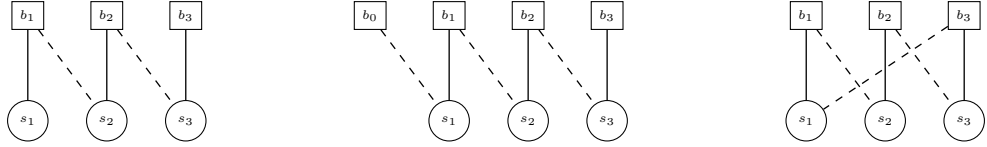

    Each $\pi$ defines a subgraph $G^\pi_\x$ that consists of buyers and stores in $\pi$. Since $\z$ and $\z'$ are matchings in $G_\x$, no two alternating paths or cycles intersect.
    We will prove $\z$ is a maximum weight matching when restricted to each such subgraph $G^\pi_\x$. This implies that $\z$ is a max weight matching in $G_x$. We define some notations that allow us to operate on each subgraph $G^\pi_\x$. Let $\z^\pi$ denote $\z$ restricted to $\pi$ and $\x^\pi$ denote $\x$ restricted to $\pi$: $$z^\pi_{bs}=\begin{cases}
        z_{bs} & \text{ if } b\in \pi, s\in \pi,\\
        0 & \text{otherwise.}
    \end{cases},\;
    x^\pi_{bsd}=\begin{cases}
        x_{bsd} & \text{ if } b\in \pi, s\in \pi,\\
        0 & \text{otherwise.}
    \end{cases}
    $$ 
    
    Similarly, let $\z^{'\pi}$ be $\z'$ restricted to $\pi$.
    Let $\z^{\backslash \pi}=\z-\z^\pi$ be the part of $\z$ not involving buyers and stores in $\pi$, and similarly $\x^{\backslash \pi}=\x-\x^\pi$ be the part of $\x$ not involving buyers and stores in $\pi$.
    
    \paragraph{Case I. Alternating path starting at $s_1$ or alternating cycle.} 
    All buyers and stores in $\pi$ are matched in $\z^\pi$ as well as $\z^{'\pi}$. This allows us to define a feasible allocation $\mathbf{x'^\pi}$ on three-sided market for the buyer allocation $\z'^{\pi}$ in the following way: A buyer matched in $\z'^{\pi}$ receives delivery from the same courier as in $\x^\pi$
    \begin{equation*}
    x'^\pi_{bsd} = \begin{cases}
        1 & \text{if } z'^{\pi}_{bs}=1 \text{ and }  \sum_s x^\pi_{bsd}=1, \\
        0 & \text{otherwise.}
    \end{cases}     
    \end{equation*}
    As $\x'^\pi$ only allocates buyers stores and couriers allocated in $\x^\pi$, $\x'^\pi+\x^{\backslash \pi}$ is another matching. Since $\mathbf{x}$ is the allocation with the optimal welfare, we have $W(x)\geq W(\x'^\pi)+W(\x^{\backslash \pi})$.
    \begin{align*}
        \sum_{bs: x^{\pi}_{bs}=1}[v_b(s)-c(b,s)]-\sum_{bd: \sum_s x^{\pi}_{bsd}=1}c_d(b) 
        &= W(\x^\pi) =W(\x)-W(\x^{\backslash \pi})
        \geq W(\x'^\pi)\\
        &= \sum_{bs: x'^{\pi}_{bs}=1}v_b(s) - \sum_{bs: x'^{\pi}_{bs}=1}c(b,s) - \sum_{bd: \sum_s x'^{\pi}_{bsd}=1}c_d(b)\\
        &\geq \sum_{bs: x'^{\pi}_{bs}=1}[v_b(s)-\ut_{bs}] - \sum_{bd: \sum_s x^\pi_{bsd}=1}c_d(b).
    \end{align*}
    The last inequality comes from tip $\ut_{bs}\geq c(b,s)$ for $s\neq s_\x(b)$, and $\mathbf{x'^\pi}$ using the same couriers to deliver for buyers in $\mathbf{x^\pi}$. Simplifying the inequality we have $$ \sum_{bs: z^{\pi}_{bs}=1}v_b(s) \geq \sum_{bs: \z'^\pi_{bs}=1}[v_b(s)-\ut_{bs}]=\sum_{bs: z'^{\pi}_{bs}=1}v^\x_b(s),$$
    which means $\z$ is a max weight matching when restricted to $G^\pi_\x$.

    \paragraph{Case II. Alternating path starting at $b_0$.} The buyer $b_t$ is matched in $\z$ but not $\z'$. 
    To denote the dependency on $\x$, let $\bar{u}_d(\x):=\bar{u}_d$ be the highest courier utility in \cref{lem:exist_courier_plan_serve} and \cref{eq:highest_minimum_tip}. 
    For an alternating path starting at $b_0$, let $d_0$ be the courier who is willing to deliver order $(b_0,s_1)$ with the lowest tip $d_0=\argmin_{d}\{c_d(b_0,s_1)+\bar{u}_d(\x)\}$. As $|\Omega_\x|<l$ \cref{lem:exist_courier_plan_serve} gives an expression for $d_0$'s highest utility $\bar{u}_{d_0}(\x)=C_{\Omega_\x}(G_D\backslash d_0)-C_{\Omega_\x}(G_D)$. 
    Here $C_{\Omega_\x}(G_D)$ is the minimum cost of delivering all orders in $\Omega_\x$, also the total courier cost for $\x$ because $\x$ is optimal.
    $C_{\Omega_\x}(G_D\backslash d_0)$ is the cost of the minimum cost matching that covers $\Omega_\x$ without courier $d_0$ in the bipartite graph $G_D=(O,D)$. Let $\mathbf{y}^{\backslash d_0}$ be the courier allocation in $G_D$ that covers $\Omega_\x$ with the minimum cost without courier $d_0$. The buyer-courier part of the cost for $\mathbf{y}^{\backslash d_0}$ is equal to total cost minus the buyer--store part of the cost $C_{\Omega_\x}(G_D\backslash d_0)-\sum_{bs:z_{bs}=1}c(b,s)$.
    
    Now define a feasible allocation $\x'$ for the buyer allocation $\z'^{\pi}+z^{\backslash \pi}$ where 
    $$x'_{bsd}=\begin{cases}
        1 & \text{ if } (b,s,d)=(b_0,s_1,d_0),\\
        1 & \text{ if } (b,s,d)\neq(b_0,s_1,d_0) \mbox{ and } z'^\pi_{bs}+z^{\backslash \pi}_{bs}=1 \mbox{ and } \sum_{s'}y^{\backslash d_0}_{o,d}\mbox{ for } o=(b,s'),\\
        0 & \text{otherwise.}
    \end{cases}$$
    
    This means, $\x'$ fulfills all orders in $z^p$ and $z^{\backslash d_0}$, by having $d_0$ delivering for $b_0$, and have the buyers receive delivery from the same couriers that they receive deliver in $\mathbf{y}^{\backslash d_0}$. The buyer-courier part of the delivery cost for $\x'$ is weakly smaller than $c_{d_0}(b_0)+C_{\Omega_\x}(G_D\backslash d_0)-\sum_{bs:z_{bs}=1}c(b,s)$. This is because an order $(b_t,s_t)$ is matched in $\mathbf{y}^{\backslash d_0}$, but $b_t$ is no longer matched in $\x'$. The welfare of $\mathbf{x'}$ can be expressed as
    \begin{align*}
        W(\mathbf{x'})&\geq \sum_{bs: z'^\pi_{bs}=1} [v_b(s)-c(b,s)] +\sum_{bs: z^{\backslash \pi}_{bs}=1} [v_b(s)-c(b,s)] - c_{d_0}(b_0)- C_{\Omega_\x}(G_D\backslash d_0)+\sum_{bs:z_{bs}=1}c(b,s)
    \end{align*}
    We can also write out the welfare of $\mathbf{x^{\backslash \pi}}$ by expressing the couriers cost indirectly through the courier cost for $\mathbf{x^\pi}$.
    \begin{align*}
        W(\mathbf{x^{\backslash \pi}})=\sum_{bs: x^{\backslash \pi}_{bs}=1} v_b(s)-[C_{\Omega_\x}(G_D)-\sum_{bs: z^{\pi}_{bs}=1}c(b,s)-\sum_{bd: \sum_s x^{\pi}_{bsd}=1}c_d(b)].
    \end{align*}
    Putting the two together we have,
    \begin{align*}
        W(\mathbf{x'})-W(\mathbf{x^{\backslash \pi}})&\geq \sum_{bs: z'^\pi_{bs}=1} [v_b(s)-c(b,s)] - [c_{d_0}(b_0)+ C_{\Omega_\x}(G_D\backslash d_0)-C_{\Omega_\x}(G_D)]\\
        & + (\sum_{bs:z_{bs}=1}c(b,s)-\sum_{bs: z^{\pi}_{bs}=1}c(b,s)- \sum_{bs: z^{\backslash \pi}_{bs}=1} c(b,s))- \sum_{bd: \sum_s x^{\pi}_{bsd}=1}c_d(b)\\
        &\geq \sum_{bs: z'^\pi_{bs}=1} [v_b(s)-c(b,s)] - (c_{d_0}(b_0)+\bar{u}_{d_0}(\x))-\sum_{bd: \sum_s x^{\pi}_{bsd}=1}c_d(b).
    \end{align*}
    
    By $\mathbf{x}$ having larger welfare than $\mathbf{x'}$,
    \begin{align*}
        \sum_{bs: z^{\pi}_{bs}=1}[v_b(s)-c(b,s)]-\sum_{bd: \sum_s x^{\pi}_{bsd}=1}c_d(b) 
        &= W(\mathbf{x}^\pi) =W(\mathbf{x})-W(\mathbf{x}^{\backslash \pi})
        \geq W(\mathbf{x}')-W(\mathbf{x}^{\backslash \pi}).
    \end{align*}
    
    Combining the latter two inequalities and simplifying,
    \begin{align*}
        \sum_{bs: z^{\pi}_{bs}=1}v_b(s) &\geq \sum_{bs: z'^\pi_{bs}=1} [v_b(s)-c(b,s)] - (c_{d_0}(b_0)+\bar{u}_{d_0}(\x)) \\
        &\geq v_{b_0}(s_1)-\ut_{b_0s_1}+ \sum_{bs: z'^\pi_{bs}=1, b\neq b_0} [v_b(s)-\ut_{bs}] = \sum_{bs: z'^{\pi}_{bs}=1}v^\x_b(s).
    \end{align*}
 
    In the last inequality we have used that for an order $(b,s)$ matched in $z'^\pi_{bs}=1$ but $z^\pi_{bs}=0$, $\ut_{bs}=\min_d (c_d(b,s)+\bar{u}_d(\x))\geq \min_d c_d(b,s)\geq c(b,s)$.
  We have proved that $\z$ is a max weight matching when restricted to $G^\pi_\x$.
\end{proof}

\subsection{Structures on buyer valuations}\label{sec:market_structures_buyers}

As an alternative to requiring some structure to courier costs, we now introduce a structure on buyer valuations.
The structure on buyer valuations is that of {\em single-mindedness}: each buyer only has positive valuation for one store. We first show that the optimal welfare can be found in polynomial time for a market with single-minded buyers, regardless of courier costs.

\begin{restatable}{theorem}{SingleMindedBuyer}
    The optimal welfare can be computed in polynomial time when each buyer only has a positive valuation for one store.
\end{restatable}
Just as \cref{thm:divisible_courier_cost_opt}, we reduce the problem of finding the optimal welfare into a minimum cost flow problem. Figure~\ref{fig:OneBuyerOneSeller} illustrates such a flow network. For each buyer $b$ that values store $s$ positively, we create a vertex $(b,s)$. For a store vertex $s$, there is an outgoing edge to each vertex $(b,s)$ where $b$ values $s$. The full reduction is similar to that in \cref{thm:divisible_courier_cost_opt} and omitted. 
    \begin{figure}[t] 
    \centering
    \begin{tikzpicture}[scale=1.3]
        \node[draw, shape=circle] (Source) at (0, 2.25) {\scriptsize So};
        \node[draw, shape=circle] (Sink) at (10, 2.25) {\scriptsize Si};
        \foreach \i/\label in {0.9/$s_1$, 3.6/$s_2$}
            \node[draw, shape=circle,minimum size=0.1cm] (\label) at (2.5, \i){\scriptsize\label};
        \foreach \i/\label in {0/$b_1 s_1$, 1.5/$b_2 s_1$, 3/$b_3 s_2$, 4.5/$b_4 s_2$}
            \node[draw, shape=circle,minimum size=0.1cm] (\label) at (5, \i){\scriptsize\label};
        \foreach \i/\label in {0.6/$d_1$, 2.25/$d_2$, 3.9/$d_3$}
            \node[draw, shape=circle,minimum size=0.1cm] (\label) at (7.5, \i){\scriptsize\label};

        \foreach \x/\y in {Source/$s_1$, Source/$s_2$}
            \draw[->] (\x) -- node[midway, right]{} (\y);
        \foreach \x/\y/\w in {$s_1$/$b_1 s_1$/$-v_{b_1}(s_1)$, $s_2$/$b_3 s_2$/$-v_{b_3}(s_2)$}
            \draw[->] (\x) -- node[below]{\w} (\y);
        \foreach \x/\y/\w in {$s_1$/$b_2 s_1$/$-v_{b_2}(s_1)$, $s_2$/$b_4 s_2$/$-v_{b_4}(s_2)$}
            \draw[->] (\x) -- node[above]{\w} (\y);

        \foreach \x/\y/\w in {$b_4 s_2$/$d_2$/$c_{d_2}^{b_4 s_2}$, $b_4 s_2$/$d_1$/$c_{d_1}^{b_4 s_2}$,$b_3 s_2$/$d_3$/$c_{d_3}^{b_3 s_2}$, $b_3 s_2$/$d_2$/$c_{d_2}^{b_3 s_2}$, $b_3 s_2$/$d_1$/$c_{d_1}^{b_3 s_2}$,
        $b_2 s_1$/$d_3$/$c_{d_3}^{b_2 s_1}$, $b_2 s_1$/$d_2$/$c_{d_2}^{b_2 s_1}$, $b_2 s_1$/$d_1$/$c_{d_1}^{b_2 s_1}$,$b_1 s_1$/$d_3$/$c_{d_3}^{b_1 s_1}$, $b_1 s_1$/$d_2$/$c_{d_2}^{b_1 s_1}$}
            \draw[->] (\x) -- node[]{}(\y);
        \foreach \x/\y/\w in {$b_4 s_2$/$d_3$/$c_{d_3}(b_4\text{,}s_2)$,$b_1 s_1$/$d_1$/$c_{d_1}(b_1\text{,}s_1)$}
            \draw[->] (\x) -- node[above]{\w} (\y);
        \foreach \x in {$d_1$,$d_2$,$d_3$}
            \draw[->] (\x) -- node[]{}(Sink);
            
        \node[] at (2.5, 5.5)  {\footnotesize Store};
        \node[] at (5, 5.5)  {\footnotesize Buyer--Store};
        \node[] at (7.5, 5.5)  {\footnotesize Courier};
    \end{tikzpicture}
    \caption{An example min-cost flow network for a market where each buyer only values one store. There is a vertex $bs$ for each buyer $b$ that values store $s$.
    \label{fig:OneBuyerOneSeller}}
    \end{figure}
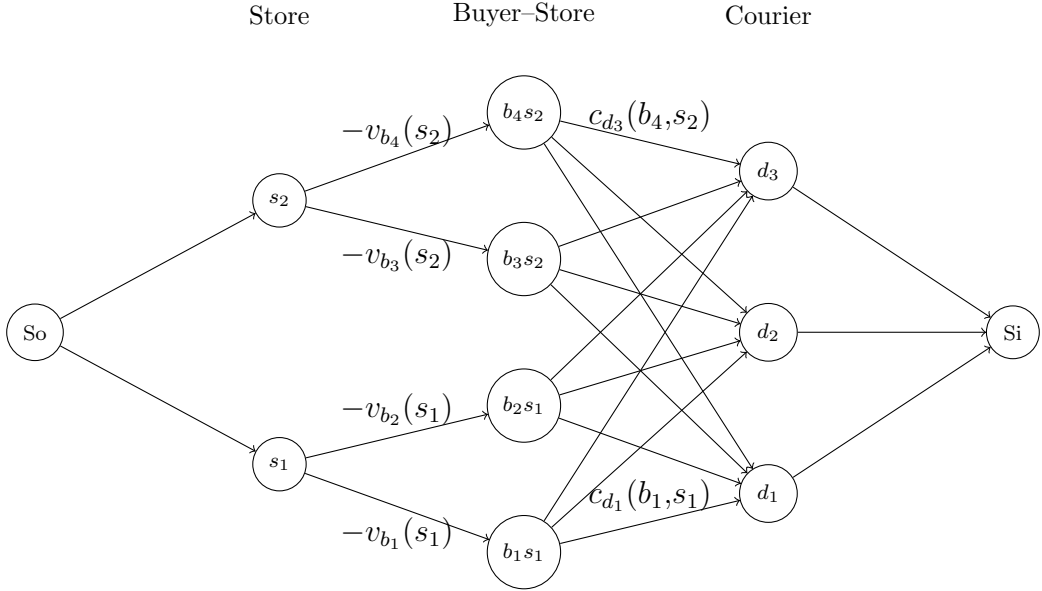

We now examine the equilibrium welfare with single-minded buyers. For without-tip equilibrium, recall \cref{example:without_tip_bad} where buyers are single-minded but the best without-tip equilibrium has welfare -1. Further with the same logic as that in \cref{fig:without_tip_not_exists}, a without-tip equilibrium may not exists in the first place when there are fewer couriers than buyers or stores, even when buyers are single-minded.

\begin{restatable}{theorem}{OneBuyerOneStore}
\label{thm:one_buyer_one_store}
    When each buyer only has a positive valuation for one store, there always exists an equilibrium that achieves the optimal welfare. However, there exists a market
    with such buyer valuations, where every without-tip equilibrium has arbitrarily low welfare.
\end{restatable}
The proof follows a similar approach as \cref{thm:divisible_courier_cost}. By assuming each buyer only values one store, we effectively limit the length of alternating path in the proof of \cref{thm:divisible_courier_cost} to three: $(b_0,s_1,b_1)$, where both $b_0$ and $b_1$ only values $s_1$ positively. The full proof is relegated to \cref{app:market_structures}.

\section{Discussion}\label{sec:dis}
We have examined the role of tips in pricing on a three-sided delivery platform. We define a with-tip equilibrium that allows buyers to offer minimum tips that still incentivize couriers to deliver, and a without-tip equilibrium where buyers are not allowed to tip. While with-tip equilibria always exist, without-tip equilibria may not exist unless there are sufficiently many couriers. Moreover, the optimal with-tip equilibrium welfare is weakly higher than the optimal without-tip equilibrium welfare, highlighting the benefits of tips. However, even with-tip equilibria can suffer from inefficiencies, and computing the optimal equilibrium welfare is NP-hard. To address these two challenges, we identify market structures that ensure the existence of efficient with-tip equilibria. Specifically, these are markets where delivery costs are decomposable into store--buyer and courier--store components, or store--buyer and buyer-courier components, as well as markets where buyers are single-minded. In any such market, an efficient with-tip equilibrium can be computed in polynomial time. 

Several interesting questions remain, offering avenues for future research:
\begin{itemize}
    \item We have assumed that stores set prices equal to their costs to turn off the inefficiency that welfare-optimal trade cannot occur when prices are higher than costs.
    However, in a free market, store prices are influenced by demand and supply, which in our model correspond to buyer orders and courier availability.
    It remains an open problem to define an equilibrium that accounts for store incentives, in setting store prices, while ensuring the  existence of an equilibrium in a three-sided market.
    \item Our model does not require budget balance and allows platforms to subsidize delivery. This is a key aspect in deriving \cref{lem:exist_courier_plan_serve,lem:exist_courier_plan_serve_equal_size}. The subsidy on delivery enables couriers to receive high utility, and forces buyers to pay high tips should they deviate from a with-tip equilibrium. It is interesting to analyze with-tip equilibrium existence and efficiency without delivery subsidy.
    \item The structure in courier costs or buyer valuations we identify are sufficient conditions to ensure the existence of efficient with-tip equilibria. What are the necessary conditions for the existence of with-tip efficient equilibria?
    \item Can we find a polynomial-time algorithm that outputs with-tip equilibrium with good welfare in every market, not just the ones with specific structures on delivery costs and buyer valuations?
    \item Our model is static, which captures well settings with time-sensitive deliveries--- buyers who do not receive their food within a short period may no longer value the delivery. It is worth considering a multi-period game, where buyers not attaining delivery for one period can wait until the platform boosts delivery compensations associated with their orders in some future periods.
\end{itemize}

\newpage
\bibliographystyle{ACM-Reference-Format}
\bibliography{sample-bibliography}

\appendix
\section{Closed-Form Solution for Minimum Tip}\label{app:min_tip}
We now give the closed form solution for $\underline{t}_{bs}$.
\begin{lemma}
    Given delivery compensation $\mathbf{w}$ and other buyers' tips $\mathbf{t}_{-b}$, buyer $b$'s minimum tip required to have some courier deliver from a store $s$ to her is given by 
$$\underline{t}_{bs}=\min_d\{\max\{0,c_d(b,s)-w_{bs},\max_{(b',s')\neq (b,s)}\{w_{b's'}+t_{b's'}-c_d(b',s')-w_{bs}+c_d(b,s)\}\}\}$$
where $t_{b's'}=0$ for $b'=b, s'\neq s$.
\end{lemma}
\begin{proof}
    We first prove that when buyer $b$ finds the tip profile $\mathbf{t}_{b}$ with the lowest $t_{bs}$. W.l.o.g., $\mathbf{t}_{bs'}=0$ for $s'\neq s$. Suppose $\exists s'\neq s, t_{bs'}>0$ and $(b,s)\in \BR_d(\mathbf{w},\mathbf{t}_b,\mathbf{t}_{-b})$ for a courier $d$. Then $w_{bs}-c_d(b,s)+t_{bs}\geq w_{bs'}-c_d(b,s')+t_{bs'}$. By changing $t_{bs'}$ to $0$, it remains that $(b,s)\in \BR_d(\mathbf{w},\mathbf{t}_b,\mathbf{t}_{-b})$.
    
    \noindent The condition for a courier $d$ satisfying $(b,s)\in \BR_d(\mathbf{w},\mathbf{t}_b,\mathbf{t}_{-b})$ is
    \begin{align*}
        w_{bs}+t_{bs}-c_{d}(b,s) &\geq w_{b's'}+t_{b',s'}-c_{d}(b',s')  \text{\; for all \;} (b',s')\neq (b,s) \text{\; and}\\
        w_{bs}+t_{bs}-c_{d}(b,s) &\geq 0 \text{\; and \;} t_{bs}\geq 0
    \end{align*}
    For $b'\neq b$, the value of $\mathbf{t}_{b'}$ is already given by $\mathbf{t}_{-b}$. And we have already proved that w.l.o.g, $t_{bs'}=0$ for $s'\neq s$. So the minimum tip required for courier $d$ to delivery from $s$ to $b$ is \begin{equation}\label{eq:min_tip_for_courier_d}
    \max\{0,c_d(b,s)-w_{bs},\max_{(b',s')\neq (b,s)}\{w_{b's'}+t_{b's'}-c_d(b',s')-w_{bs}+c_d(b,s)\}\}
    \end{equation}
\end{proof}

\section{Missing Proofs and Lemmas in \cref{sec:motivate_use_tip}}\label{app:eq_existence}

We prove \cref{lem:exist_courier_plan_serve} and \cref{lem:exist_courier_plan_serve_equal_size} together. 
\existscourierPlan*
\existscourierPlanEqualSize*
\begin{proof}
    Given $\Omega$, define a two-sided market $M=(\Omega, D)$ where $\Omega$ are consumers and $D$ the items. In $M$, a consumer $o=(b,s)$ values item $d\in D$ at $v_{o}(d)=H-c_d(o)=H-c_d(b,s)$, where $H=\sum_{bs}v_b(s)+\sum_{bs}w_{bs}+\sum_{bsd}c_d(b,s)$. We build a one-to-one correspondence between courier plans that serves $\Omega$ to walrasian equilibria in $M$.

    Given $(\mathbf{w},\mathbf{y})$, let $\mathbf{y}'$ be the allocation of items $D$ to consumers $\Omega$ that follow $\mathbf{y}$ (i.e., $\forall o\in\Omega, d\in D, y'_{od}=y_{od}$). Let $\mathbf{u}'$ be the vector of courier utility (i.e., for $o_\mathbf{y}(d)=(b,s), u'_d=u_d(o_\mathbf{y}(d))=w_{bs}-c_d(o)$). We prove if $(\mathbf{w},\mathbf{y})$ serves $\Omega$, then $(\mathbf{y}',\mathbf{u}')$ is a walrasian equilibrium in $M$, where $\mathbf{y}'$ is the walrasian allocation and $\mathbf{u}'$ the walrasian price.

    First, an item $d$ in $M$ that is not allocated corresponds to a courier $d$ who do not deliver, and have $u'_d=0$. Second, each consumer $o$ in $M$ receives the favorite item. Consider a consumer $o=(b,s)$ receiving item $d\neq \emptyset$ with utility $H-c_d(b,s)-u'_d=H-w_{bs}$. If $o$ buys from any other item $d'$ which is not allocated $u'_{d'}=0$, she has utility $H-c_{d'}(b,s)-u'_{d'}=H-c_{d'}(bs)$. But being a courier plan, $(\mathbf{w},\mathbf{y})$ satisfies for courier $d': w_{bs}-c_{d'}(bs)\leq 0$. This means consumer $o$ has no incentive to deviate to some item $d'$ unallocated. Alternatively, if $o$ buys from any other item $d'$ who is allocated to a consumer $o'=(b',s')\neq (b,s)$, she has deviation utility $H-c_{d'}(b,s)-u'_{d'}=H-c_{d'}(b,s)-w_{b',s'}+c_{d'}(b',s')$. Being a courier plan, $(\mathbf{w},\mathbf{y})$ satisfies for courier $d': w_{b's'}-c_{d'}(b',s')\geq w_{bs}-c_{d'}(b,s)$, and  there is no incentive for consumer $o$ to deviate to buy item $d'$. Finally customer $o$ has weakly positive utility $H-w_{bs}>0$.

    We now look at the other direction. Given a walrasian equilibrium $(\mathbf{y}',\mathbf{u}')$ in $M$, we construct a courier plan that serves $\Omega$. As $|\Omega|\leq l$ and each customer $o$ has valuation $v_o(d)=H-c_d(o)>0$ towards any item $d$, by the first welfare theorem, all customers $o\in \Omega$ are allocated to some courier. To construct the courier plan $(\mathbf{w},\mathbf{y})$, for $o=(b,s)\notin \Omega$, set $w_{b,s}=0, y_{od}=0, \forall d$. For $o=(b,s)\in\Omega$ where in the walrasian equilibrium $o$ buys an item $d$, set $\forall d', y_{od'}=y'_{od'}$ and $w_{b,s}=u'_d+c_d(b,s)$. This construction guarantees that $(\mathbf{w},\mathbf{y})$ serves $\Omega$ if it is a courier plan. We now prove that if $(\mathbf{y}',\mathbf{u}')$ is a walrasian equilibrium in $M$, then $(\mathbf{w},\mathbf{y})$ is a courier plan.
    
    For any courier $d$ that delivers $o_\mathbf{y}(d)=(b,s)\neq \emptyset$, her utility is $w_{b,s}-c_d(b,s)=u'_d\geq 0$ because walrasian equilibrium $(\mathbf{y}',\mathbf{u}')$ satisfies all prices are weakly less than zero. Deviating to a store $o'\notin \Omega$ yields weakly less than zero utility. Deviating to a store $o'=(b',s')\in\Omega$, yields utility  $w_{b',s'}-c_d(b',s')$. But since every customer $o'\in \Omega$ are allocated an item $d'$ in the Walrasian equilibrium, $w_{b',s'}-c_d(b',s')=u'_{d'}+c_{d'}(b',s')-c_d(b',s')$. Since walrasian equilibrium $(\mathbf{y}',\mathbf{u}')$ guarantees customer $o'$ incentives, $H-c_{d'}(b',s')-u'_{d'}\geq H-c_{d}(b',s')-u'_d \Leftrightarrow u'_d\geq u'_{d'}+c_{d'}(b',s')-c_{d}(b',s')$. This means courier $d$ does not want to deviate to the store $o'$. For $d, o_\mathbf{y}(d)\in \BR_d(\mathbf{w})$.

    For any courier $d$ that does not deliver in $\mathbf{y}$, delivering an order $o'\notin \Omega$ yields utility weakly less than zero. Delivering an order $o'=(b',s')\in \Omega$, which is delivered by another courier $d'$ yields utility $w_{b',s'}-c_d(b',s')=u'_{d'}+c_{d'}(b',s')-c_d(b',s')$. But examining the incentive of customer $(b',s')$ in the walrasian equilibrium, $H-c_{d'}(b',s')-u'_{d'}\geq H - c_d(b',s')-0\Leftrightarrow u'_{d'}+c_{d'}(b',s')-c_d(b',s')\leq 0$. So $d$ does not deviate to $o'$. For $d, o_\mathbf{y}(d)\in \BR_d(\mathbf{w})$.

    We have established an one-to-one relationship between courier plans that serves $\Omega$ and walrasian equilibria in $M$. Since a walrasian equilibrium always exists for unit-demand supply market, there is always a courier plan that serves $\Omega$. By the first welfare theorem, walrasian equilibrium in $M$ achieves the optimal welfare, so every $o\in \Omega$ must be allocated in $\mathbf{y}$. Then in a courier plan, $\mathbf{y}$ must correspond to a minimum cost matching that covers all $\Omega$. This is because (i) edge weight in $M$ is defined as $H$ minus cost, (ii) $\mathbf{y}$ serves $\Omega$.
    
    We now upper bound the utility of couriers in the courier plan that serves $\Omega$. Use $\mathbf{y}^{'*}$ to denote the max weight matching of the market $M$ with $SW(M)$ being its weight. Use $\mathbf{y}^{'*}_{-d}$ to denote the max weight matching of the market $M$ without item $d$, with $SW(M\backslash d)$ being its weight. If a courier in has $o_\mathbf{y}(d)=\emptyset$, then $u_d(o_\mathbf{y}(d))=0$. If $d$ is matched to $(b,s)\in\Omega$, its utility satisfies
    \begin{align*}
        u_d(o_\mathbf{y}(d))=w_{bs}-c_d(b,s)=u'_d\leq SW(M)-SW(M\backslash d)
    \end{align*}
    The inequality is the highest achievable walrasian equilibrium price, proved in \cite{gul1999walrasian}. We can define the highest possible courier compensation $\mathbf{w}$ with this highest walrasian equilibrium price $\mathbf{u}'$. 
    When $l>|\Omega|$, it satisfies that the size of $\mathbf{y}^{'*}$ is equal to the size of $\mathbf{y}^{'*}_{-d}$, and all customers in $\Omega$ in market $M$ are matched for $\mathbf{y}^{'*}_{-d}$.
    \begin{align*}
         SW(M)-SW(M\backslash d) &= \sum_{(o,d)\in \mathbf{y}^{'*}}[H-c_d(o)] - \sum_{(o,d)\in \mathbf{y}^{'*}_{-d}}[H-c_d(o)]\\
         &=\sum_{(o,d)\in \mathbf{y}^{'*}_{-d}}c_d(o) - \sum_{(o,d)\in \mathbf{y}^{'*}}c_d(o)\\
         &= C_{\Omega}(G_D\backslash d)-C_{\Omega}(G_D)
    \end{align*}
    where $C_{\Omega}(G_D)$ is the cost of the minimum cost matching that covers $\Omega$ in $G_D$, and $C_{\Omega}(G_D\backslash d)$ is the cost of the minimum-cost matching that covers $\Omega$ without courier $d$. 
    
    When $l=|\Omega|$, the size of $\mathbf{y}^{'*}$ is larger than the size of $\mathbf{y}^{'*}_{-d}$ by 1. So $\mathbf{y}^{'*}_{-d}$ is the maximum weight matching that covers all but one customer in $\Omega$ in $M$. 
    \begin{align*}
         SW(M)-SW(M\backslash d) &= \sum_{(o,d)\in \mathbf{y}^{'*}}[H-c_d(o)] - \sum_{(o,d)\in \mathbf{y}^{'*}_{-d}}[H-c_d(o)]\\
         &=H+\sum_{(o,d)\in \mathbf{y}^{'*}_{-d}}c_d(o) - \sum_{(o,d)\in \mathbf{y}^{'*}}c_d(o)
         \geq \max_{b,s}v_b(s)
    \end{align*}
\end{proof}

\paretoDominant*
\begin{proof}
    As $\mathbf{t}=0$ and $(\mathbf{p},\mathbf{w},\mathbf{x})$ is a without-tip equilibrium, it satisfies that stores not bought from have zero purchase price, and orders not delivered have zero delivery compensation and zero tips $t_{bs}$=0. For couriers, as tips are zero, the set of orders that maximize utility with tips is the same as that without tips.
    
    Consider a buyer $b$ buying from a store $s\neq \emptyset$ in $\mathbf{x}$. As $\mathbf{t}_{-b}=0$ does not change couriers incentives, the minimum tip to have some courier delivers from store $s$ to her is zero $\underline{t}_{bs}=0$. So buying from $s$ remains the best option for $b$, i.e., $\forall s',\; v_b(s)-p_s-\underline{t}_{bs}=v_b(s)-p_s\geq v_b(s')-p_{s'}\geq v_b(s')-p_{s'}-\underline{t}_{bs'}$ where the first inequality is because $s\in \BR_d(\mathbf{p})$. A buyer $b$ not buying in $\mathbf{x}$ still does not buy now as purchase prices do not change.
\end{proof}

\ZeroTipEquivalence*
\begin{proof}
    For a store $s$ that is allocated in $\mathbf{x}$ to a buyer $b$, set $p'_{s}=p_s+t_{bs}$. For a store $s$ that is not allocated, set $p'_{s}=0$. For an order $(b,s)$ delivered by a courier $d$ in $\mathbf{x}$, set $w'_{bs}=w_{bs}+t_{bs}$. For an order $(b,s)$ not delivered by any courier in $\mathbf{x}$, set $w'_{bs}=0$. From the couriers' perspective, the delivery compensation plus tips associated with each order is the same: $w_{bs}+t_{bs}=w'_{bs}$, so the set of orders that maximizes each courier utility remains unchanged. We only need to prove that $s_\mathbf{x}(b)$ is still in the set of stores that maximizes each buyer utility in $(\mathbf{p}',\mathbf{w}',\mathbf{t}',\mathbf{x})$. 
    
    Consider a buyer $b$ buying from a store $s_\mathbf{x}(b)$. Let $\ut_{bs}(d)$ and $\ut'_{bs}(d)$ be the minimum tip required for a courier $d$ to deliver from store $s\neq s_\mathbf{x}(b)$ to $b$, under $(\mathbf{p},\mathbf{w},\mathbf{t},\mathbf{x})$ and $(\mathbf{p}',\mathbf{w}',\mathbf{t}',\mathbf{x})$ respectively. We first show that $\ut_{bs}(d)=\ut'_{bs}(d)$. According to \cref{eq:min_tip_for_courier_d}, their values are given by 
    \begin{align*}
        \ut_{bs}(d) &= \max\{0,c_d(b,s)-w_{bs},\max_{(b',s')\neq (b,s)}\{w_{b's'}+t_{b's'}-c_d(b',s')-w_{bs}+c_d(b,s)\}\}\\
        &= \max\{0,c_d(b,s),\max_{(b',s')\neq (b,s)}\{w_{b's'}+t_{b's'}-c_d(b',s')+c_d(b,s)\}\}\\
        \ut'_{bs}(d) &= \max\{0,c_d(b,s)-w'_{bs},\max_{(b',s')\neq (b,s)}\{w'_{b's'}+t'_{b's'}-c_d(b',s')-w'_{bs}+c_d(b,s)\}\}\\
        &= \max\{0,c_d(b,s),\max_{(b',s')\neq (b,s)}\{w'_{b's'}-c_d(b',s')+c_d(b,s)\}\}
    \end{align*}
    Here we have used the condition that $w_{bs}=w'_{bs}=0$ for $s\neq s_\mathbf{x}(b)$ in equilibrium, and $\mathbf{t}'=0$. For an order $(b',s')\neq(b,s)$, by construction it always satisfies that $w'_{b's'}=w_{b's'}+t_{b's'}$. So $\ut_{bs}(d)=\ut'_{bs}(d)$ and the minimum tip to get a delivery from store $s$ is the same $$\ut_{bs}=\min_d \ut_{bs}(d)= \min_d \ut'_{bs}(d)=\ut'_{bs}$$

    Consider a buyer $b$ which buys from store $s=s_\mathbf{x}(b)$ in $\mathbf{x}$. Now we are ready to prove that $s$ is still in the set of stores that maximizes buyer $b$'s utility in $(\mathbf{p}',\mathbf{w}',\mathbf{t}',\mathbf{x})$.
    
    When $s\neq \emptyset$, let $u_b(s')=v_b(s')-p_s-\ut_{bs'}$ and $u'_b(s')=v_b(s')-p'_{s'}-\ut'_{bs'}$ denote the utility of $b$ when buying from a store $s'$ in $(\mathbf{p},\mathbf{w},\mathbf{t},\mathbf{x})$ and $(\mathbf{p}',\mathbf{w}',\mathbf{t}',\mathbf{x})$ respectively. 
     As $(\mathbf{p},\mathbf{w},\mathbf{t},\mathbf{x})$ is a with-tip equilibrium, buyer $b$ pays the minimum tip $t_{bs}=\ut_{bs}$. It satisfies for buyer $b$ 
    and any store $s'$ that
    \begin{align*}
        u_b(s)=v_b(s) - p_s - t_{bs}=v_b(s) - p_s - \ut_{bs}\geq v_b(s') - p_{s'} - \ut_{bs'}=u_b(s')
    \end{align*}

    When $s=\emptyset$, its utility satisfies $\forall s, 0 \geq u_{b}(s)=v_b(s)-p_s-\ut_{bs}$. Now in $(\mathbf{p}',\mathbf{w}',\mathbf{t}',\mathbf{x})$, buyer $b$ when buying from a store $s$ has utility
    \begin{align*}
        u'_b(s) &= v_b(s)-p'_s-\ut'_{bs}=v_b(s)-p'_s-\ut_{bs}\\
        &\leq v_b(s)-p_s-\ut_{bs}\leq u_b(s)\leq 0
    \end{align*} 
\end{proof}

\subsection{Pricing with tips when equilibrium do not clear the market}\label{app:tips_eq_not_clear_market}
In this section, we demonstrate the benefit of tips when equilibrium are not required to clear the market. That is, unsold stores can have positive price and undelivered orders can have positive compensations and tips. Take the market in \cref{example:without_tip_bad}, the best without-tip equilibrium has welfare 0, when no transactions take place. In the contrary, the best with-tip equilibrium achieves the optimal welfare of 3. 
\begin{proposition}
    When equilibrium do not need to clear the market, there always exists a with-tip equilibrium whose welfare is at least $\mathit{OPT}/\min\{m,n,l\}$.
\end{proposition}
\begin{proof}
    Let $\mathbf{x}^\star$ be an allocation that achieves the optimal welfare $\mathit{OPT}$. In $\mathbf{x}^\star$ there are at most $\min\{m,n,l\}$ transacting buyer--store--courier triples, or at most $\min\{m,n,l\}$ trades. So there the trade that generates the largest welfare is guaranteed to generate welfare at least $\frac{1}{\min\{m,n,l\}}\mathit{OPT}$. Denote this trade by $(b^\star, s^\star, d^\star)$. We show the allocation where only $(b^\star, s^\star, d^\star)$ trades is always in a with-tip equilibrium with zero tips $\mathbf{t}=0$.
    
    Given the order $(b^\star, s^\star)$, $d^\star$ is the courier with the lowest cost for delivering $(b^\star, s^\star)$. Hence, set courier compensation $\mathbf{w}$ to be  $w_{bs}=0, \forall (b,s)\neq (b^\star,s^\star)$ and $w_{b^\star,s^\star}=c_{d^\star}(b^\star, s^\star)$. This way, no other couriers have incentives to serve $(b^\star,s^\star)$.

    For purchase prices, set $p_s=\max_{b,s}v_b(s)$ for all $s\neq s^\star$ and $p_{s^\star}=v_{b^\star} (s^\star)$. Buyer $b^\star$ do not need to pay any tips for the delivery $\underline{t}_{b^\star s^\star}=0$.  This way $b^\star$ is indifferent to buy from $s^\star$, while all other buyers $b$ have to pay $p_{s^\star}+\underline{t}_{b s^\star}$ to buy from $s$. Furthermore, for any $b\neq b^\star$ the minimum tip for some couriers to deliver $s^\star$ satisfies $\underline{t}_{b s^\star}+w_{bs^\star}=\underline{t}_{b s^\star}\geq \min_d{c_d(b,s^\star)}$. So for any other buyer $b$, purchasing from $s^\star$ results in utility $v_b(s^\star)-\min_d{c^d_{b s^\star}}-v_b^\star(s^\star)$. But since $(b^\star, v^\star, s^\star)$ is the trade with the largest welfare, $v_{b^\star} (s^\star)\geq v_{b^\star}(s^\star)-c_{d^\star}(b^\star,s^\star)\geq v_{b}(s^\star)-c_d(b,s^\star) \; \forall b,d$. So any other buyer $b$ will have a weakly smaller than zero of utility of purchasing and paying tips for $s^\star$. So buyers incentives are satisfied.
\end{proof}

\section{Missing Proofs and Lemmas in \cref{sec:general_markets}}\label{app:general_markets}

\OPTHard*
\begin{proof}
    We reduce from the 3DM, one of the first 21 NP-complete problems proved \citep{karp1975computational}. 3DM asks the following question: Given a 3-regular hypergraph $G=(B,S,D)$ with hyperedges $T\subset B \times S \times D$ where $|B|=|S|=|D|=n$, is there a perfect matching? 
    
    we construct a following market $M$ with $n$ buyers, stores and couriers. Buyers' valuations are always 1. Couriers' costs are set as 
    \begin{equation}
    c_d(b,s) =
    \begin{cases}
      0 & \text{if $ (b,s,d)\in T$,}\\
      1 & \text{otherwise.}
    \end{cases}       
    \end{equation}
    If $G$ has a perfect matching of size $n$, then $M$ has an allocation of welfare $n$ defined by the perfect matching. If $M$ has an allocation $\mathbf{x}$ of welfare $n$, for each $(b,s,d)$ where $x_{bsd}=1$, couriers' cost must be zero $c_d(bs)=0$. So there exists a hyperedge $(b,s,d)\in T$. Furthermore, all buyer, stores and couriers are involved. This allocation corresponds to a perfect matching in $G$.
\end{proof}

\highestcourierPriceEq*
\begin{proof}
    Let $\Omega_\mathbf{x}=\{(b,s)|\sum_d x_{bsd}=1\}$ and $y$ a courier allocation defined by 
    $$ y_{od} =\begin{cases}
    1 & o=(b,s)\in \Omega \text{ and } x_{bsd}=1,\\
    0 & \text{otherwise.}
    \end{cases} $$
    
    For $(\p,\w,\bt,\x)$ to be an equilibrium, $(\w,\mathbf{y})$ must be a courier plan that serves $\Omega$. \cref{lem:exist_courier_plan_serve} and \cref{lem:exist_courier_plan_serve_equal_size} state the existence of a courier plan $(\bar{\w},\mathbf{y}')$. These two lemmas also require that both $\mathbf{y}$ and $\mathbf{y}'$ are a minimum-cost matching that covers $\Omega$ in the bipartite graph $G_D=(O,D)$. Then $(\bar{\w}, \mathbf{y})$ is a courier plan that serves $\Omega$.\footnote{The proof for \cref{lem:exist_courier_plan_serve} and \cref{lem:exist_courier_plan_serve_equal_size} build a one-to-one correspondence between courier plans and Walrasian equilibrium in a two-sided market $(\Omega,D)$. $\mathbf{y}$ and $\mathbf{y}'$ corresponds to the Walrasian allocation, and $\w,\bar{\w}$ corresponds to the Walrasian price. By the second welfare theorem, $(\mathbf{y},\bar{\w})$ corresponds to a Walrasian equilibrium, so is a courier plan as well.} This shows that courier incentives are satisfied.  

    We now check buyer incentives. Consider a buyer $b$ buying from store $s_\x(b)$. As $\bt=0$ and the courier plan guarantees couriers to deliver for $(b,s_\mathbf{x}(b))$, the minimum tip buyer $b$ offers to get delivery from $s_\mathbf{x}(b)$ is zero. For buying from another store $s\neq s_\mathbf{x}(b)$, \cref{eq:highest_minimum_tip} gives the minimum tip required to get a delivery 
    $$\ut_{bs}=\min_d c_d(b,s)+\bar{u}_d$$
    When $|\Omega|=l$, it satisfies that for all couriers $\bar{u}_d>\max_{bs}\{v_b(s)\}$. So buyers utility buying from any other $s$ is negative. When $|\Omega|<l$, $\bar{u}_d$ is the highest possible courier utility over all courier plans. So the minimum tip $\ut_{bs}$ required for buyer $b$ to get delivery from store $s\neq s_\x(b)$ is the highest possible in $(\bar{\w},\mathbf{y})$. As buyers do not deviate in $(\p,\w,\bt,\x)$, she does not in $(\p,\bar{\w},\bt,\x)$ either.
\end{proof}

\testAlloc*
\begin{proof}
    The only if direction. When $\x$ is in some equilibrium, we have proved that it is always in an equilibrium of the form $(\p,\bar{\w},\bt,\x)$ where $\bt=0$. The minimum tip for buyer $b$ to buy from store $s$ it $\ut_{bs}$. From a buyer $b$'s perspective, her realized valuation for store $s$ is exactly $v^x_{bs}$, as she also has to pay tips to get the delivery. In equilibrium $(\p,\w,\bt,\x)$ all buyers buy from their favorite store, and stores unsold has price zero. These two conditions correspond to a Walrasian equilibrium in the market defined by $G_\x$. By the first welfare theorem, $\mathbf{z}$ must be a maximum weight matching in $G_\x$. 

    The if direction. $\mathbf{z}$ being a maximum weight matching defines a Walrasian equilibrium $(\mathbf{z},\p)$ in the two-sided market $G_x$, where $p_s=0$ for unallocated stores. Then $(\p,\bar{\w},\bt,\x)$ is an equilibrium, where $\bt=0$. This is because we have already shown that $\bar{\w}$, zero tips, and $\x$ satisfy courier incentives. Buyer incentives are satisfied by the Walrasian equilibrium in $G_x$. 
\end{proof}

\allcouriersDeliver*
\begin{proof}
    \cref{lem:exist_courier_plan_serve_equal_size} shows with courier compensation $\w$, any courier utility is larger than the maximum buyer valuations $\bar{u}_d >\min_{bs}v_b(s)$. This means the minimum tip $\ut_{bs}$ in  \cref{eq:highest_minimum_tip} for $(b,s)\notin\Omega_\x$ is larger than buyer valuation. Thus any edge not in $\mathbf{z}$ has negative weights and any edge in $\mathbf{z}$ has weakly positive weights. By \cref{lem:test_alloc}, $\x$ is in some equilibrium.
\end{proof}

\PolyCheck*
\begin{proof}
    For any feasible $\mathbf{x}$, we can calculate the 
    weight of edges in graph $G_\x$ to invoke \cref{lem:test_alloc}. Following their expressions in \cref{lem:exist_courier_plan_serve} and \cref{lem:exist_courier_plan_serve_equal_size}, calculating $\bar{u}_d$ for all couriers takes polynomial time, as each $\bar{u}_d$ requires two computations for minimum cost matching that covers $\Omega_\x$ or all but one element in $\Omega_\x$. Calculating the minimum tip $\ut_{bs}$ for all buyers $b$ and stores $s$, and finding the value of max weight matching in $G_\x$ takes polynomial time as well.
\end{proof}

\section{Missing Proofs in \cref{sec:market_structures}}\label{app:market_structures}
\divisiblecourierCostOpt*
\begin{proof}
    We reduce the problem of finding the optimal welfare to a standard \emph{minimum cost flow} problem, solvable in polynomial time using existing algorithms. We demonstrate this for $c_d(b,s)=c(b,s)+c_d(s)$, while the similar case $c_d(b,s)=c(b,s)+c_d(b)$ is omitted. 
    
    Given a market, construct a flow network. Figure~\ref{fig:OPTDivisibleCost} illustrates such a network for a market with two buyers, two stores and three couriers. There is a source node $S_o$, a sink node $S_i$, a node for each buyer $b\in B$, order $o\in B\times S$, store $s\in S$, courier $d\in D$, and an additional dummy node for each store. There exists a directed edge from the source to each buyer node, from each buyer to each order that involves the buyer, from each order to the store that is involved in the order, from each store to the dummy node of the same store, from each dummy store to each courier, and from each courier to the sink. All edges have capacity one. An edge between a buyer $b$ and an order $(b,s)$ has cost $-v_b(s)$, an edge between an order $(b,s)$ and a store $s$ has cost $c(b,s)$, an edge between a dummy store $s'$ and a courier have $c_d(s)$ as cost. The net supply or demand are zero for all nodes except the source and the sink node.

    Each integer flow corresponds to a feasible allocation. To see this, an integer flow ensures unit-demand, unit-capacity, and unit-supply constraints, as each buyer node receives at most one unit of flow from the source, and each store node send out one unit of flow to its dummy node, and each courier node send out one unit of flow to the sink. The path each unit of flow goes through defines a trade $x_{bsd}=1$.
    
    To solve for the optimal welfare, set net supply of the source node equal to the net demand of the sink node to be $f=1,2,...,\min\{m,n,l\}$. Solve the minimum cost flow problem for each $f$ and take the flow with the minimum cost over all $f$. By the integrality theorem, any minimum cost network flow problem instance whose demands data are all integers has an optimal solution with integer flow on each edge. 
    Each flow corresponds to a feasible allocation $\mathbf{x}$, and the cost of a flow equals to the negation of the welfare of $\mathbf{x}$. Hence finding the minimum cost flow among all $f=1,2,...,\min\{m,n,l\}$ is finding the allocation with the optimal welfare.
\end{proof}

\DivisiblecourierCost*
\begin{proof}
    We now continue the proof for $c_d(b,s)=c(b,s)+c_d(s)$ when not all stores are bought from $|\Omega_\x|<n$. The main difference from the previous case $|\Omega_\x|=n$ is that alternative paths can include stores not bought from in $\x$ now. But the analysis are similar.
    
    The two matching $\z$ and $\z'$ define some alternating paths and cycles. The proof for the case $|\Omega_\x|=n$ already works for alternative paths and cycles that do not contain unmatched stores in $\z$. So we can focus on alternating paths that involve stores that are not matched in $\z$. Stores unmatched in $\z$ cannot exist in any alternating cycles. An alternating path $p$ is denoted as 
    $$((b_0),s_1,b_1,s_2,b_2,\ldots, b_{t-1},s_t),$$ where for every $q\in \{1,\ldots, t-1\}$, $$z_{b_q,s_q}=1\mbox{ and }z'_{b_q,s_q}=0,$$ and for every $q\in \{1,\ldots, t\}$,  $$z_{b_{q-1},s_q}=0\mbox{ and } z'_{b_{q-1},s_q}=1.$$ Per discussion of the previous paragraph, an alternating path always end at some store $s_t$ that is not matched by $\z$. An alternating path can either start from 1) $b_0$, a buyer allocated by $\z'$ but not by $\z$; 2) or $s_1$, a store allocated by $\z$ but not by $\z'$. 

    Similar as the proof for $\Omega_n$, define $\z^\pi$ to be the buyer--store allocation of $\z$ when restricted to the alternating path $\pi$, and $\z^{\backslash \pi}=\z-\z^{\pi}$ the buyer--store allocation of $\z$ not involving buyers or stores in $\pi$. Let $\mathbf{x}^\pi$ and $\mathbf{x}^{\backslash \pi}$ be the part of $\mathbf{x}$ that does or does not involve $\pi$ respectively. Similarly, define $\mathbf{z}^{'\pi}$ to be $\z'$ restricted to the alternating path $\pi$. 

    \paragraph{Case I. alternating path starting at $s_1$.} This case is entirely the same as Case I. of $|\Omega_\x|=n$.
    
    \paragraph{Case II. alternating path starting at $b_0$.} If the alternative path starts from $b_0$, let $d_0$ be the courier who are willing to deliver the order $(b_0,s_1)$ with the lowest tip $d_0=\argmin_d{c_d(b_{0},s_1)+\bar{u}_d(\mathbf{x})}$. Use $C_{\Omega_\x}(G_D)$ and $C_{\Omega_\x}(G_D\backslash d_0)$ to denote the cost of minimum cost matching in $G_D$ that covers $\Omega_\x$ with and without $d_0$ respectively. As $\x$ is the welfare-optimal allocation, $C_{\Omega_\x}(G_D)$ is also the total courier cost for $\x$. Let $\mathbf{y}^{\backslash d_0}$ be the minimum cost matching in $G_D$ that covers $\Omega_\x$ without $d_0$. 
    
    Let $\mathbf{x'}$ be the minimum-cost allocation that satisfies $\sum_d x'_{bs}= z^{'\pi}_{bs}+z^{\backslash \pi}_{bs}$. To deliver the orders in $ \z^{'\pi}+\z^{\backslash \pi}$, one possible way is to have courier $d_0$ deliver $(b_{0},s_1)$ and all other buyers' orders being delivered by the same couriers as in $\mathbf{y}^{\backslash d_0}$, that is, a buyer $b$ receives delivery from a courier $d$ that satisfies $y_{od}^{\backslash d_0}=1$ for $o=(b,s)\in \Omega_x$.
    The courier--buyer part of the cost of this matching is $c_{d_0}(b_1)+[C_{\Omega_\x}(G_D\backslash d_0)-\sum_{bs:z_{bs}=1}c(b,s)]$. So total courier cost in $\mathbf{x'}$ is weakly smaller than the above stated way of matching couriers.
    \begin{align*}
        W(\mathbf{x'}) \geq \sum_{bs:z_{bs}^{'\pi}=1}[v_b(s)-c(b,s)]+\sum_{bs:z_{bs}^{\backslash \pi}=1}[v_b(s)-c(b,s)] - c_{d_0}(b_1) -C_{\Omega_\x}(G_D\backslash d_0)+\sum_{bs:z_{bs}=1}c(b,s)
    \end{align*} 
    We also write out the total welfare of $W(\mathbf{x}^{\backslash \pi})$ by expressing the courier cost indirectly through the courier cost of $\mathbf{x}^\pi$.
    \begin{align*}
        W(\mathbf{x}^{\backslash \pi}) = \sum_{bs:z_{bs}^{\backslash \pi}}v_{b}(s) - [C_{\Omega_\x}(G_D)-\sum_{bs:z^\pi_{bs}=1}c(b,s)-\sum_{bd:\sum_s x^p_{bsd}=1}c_d(b)]
    \end{align*}
    Putting the two inequalities together
    \begin{align*}
        W(\mathbf{x}') - W(\mathbf{x}^{\backslash \pi}) & \geq \sum_{bs:z_{bs}^{'\pi}=1}[v_b(s)-c(b,s)] - c_{d_0}(s_1) - [C_{\Omega_\x}(G_D\backslash d_0)-C_{\Omega_\x}(G_D)]-\sum_{bd:\sum_d x^\pi_{bsd}=1}c_d(b)
    \end{align*}
    By $\mathbf{x}$ having larger welfare than $\mathbf{x'}$
    \begin{align*}
        \sum_{bs:z^\pi_{bs}=1}[v_b(s)-c(b,s)]-\sum_{bd:\sum_s x^\pi_{bsd}=1}c_d(b)=W(\mathbf{x}^\pi)=W(\mathbf{x})- W(\mathbf{x}^{\backslash \pi}) \geq W(\mathbf{x}') - W(\mathbf{x}^{\backslash \pi})
    \end{align*}
    Combining the last two inequalities and simplifying
    \begin{align*}
        \sum_{bs:z^\pi_{bs}=1}v_b(s) &\geq \sum_{bs:z^\pi_{bs}=1}[v_b(s)-c(b,s)] \geq \sum_{bs:z_{bs}^{'\pi}=1}[v_b(s)-c(b,s)] - c_{d_0}(s_1)  - \bar{u}_{d_0}(\x)\\
        &\geq \sum_{bs:z_{bs}^{'\pi}=1}[v_b(s)-\ut_{bs}] = \sum_{bs:z_{bs}^{'\pi}=1} v^\x_b(s)
    \end{align*}
    In the last inequality, we used $\ut_{bs}\geq c(b,s)$ for $(b,s)\notin \Omega_\x$. So for any alternating path $p$ that starts from $b_0$, $\z^\pi$ have weakly larger welfare than any other matching $\z'^{\pi}$ on $p$.

    The proof for courier costs being divisible to buyer--store and store-courier part $c^d(b,s)=c_d(s)+c(b,s)$ is quite similar to the proof above for $c^d(b,s)=c_d(b)+c(b,s)$. One just goes through all four types and alternating paths and cycles, and swapping out the buyer-courier cost to store-courier cost in the proof. So we omit that proof. 
\end{proof}

\OneBuyerOneStore*
\begin{proof}
    Let $\x$ be a feasible allocation that achieves the optimal welfare $W(\x)=\mathit{OPT}$. If $|\Omega_\x|=l$, \cref{cor:all_courier_deliver} shows $\x$ is in some equilibrium. We focus on the case where $|\Omega_\x|<l$. Let $\z$ be the buyer allocation induced by $
    \x$ in $G_\x$. We prove that $\z$ is a maximum weight matching in $G_\x$, and by \cref{lem:test_alloc}, $\mathbf{x}$ is indeed in an equilibrium.

    For a store $s_0$ that is valued by some buyer $b_0$ but not matched in $\x$, let $d_0$ be the courier who is willing to deliver order $(b_0,s_0)$ with the lowest tip $d_0=\argmin_{d}\{c_{d'}(b_0,s_0)+\bar{u}_{d}(\x)\}$. By \cref{lem:exist_courier_plan_serve}, courier $d_0$'s utility is written as $\bar{u}_{d_0}(\x)=C_{\Omega_\x}(G_D\backslash d_0)-C_{\Omega_\x}(G_D)$. Here $C_{\Omega_\x}(G_D)$ is the minimum cost of delivering all orders in $\Omega_\x$, also the total courier cost for $\x$ because $\x$ is optimal.
    $C_{\Omega_\x}(G_D\backslash d_0)$ is the cost of the minimum cost matching that covers $\Omega_\x$ without courier $d_0$ in the bipartite graph $G_D=(O,D)$. Let $\mathbf{y}^{\backslash d_0}$ be the courier allocation in $G_D$ that covers $\Omega_\x$ with the minimum cost without courier $d_0$.
    
    Consider an allocation where $d_0$ delivers the order $(b_0,s_0)$, and all buyer--store pairs in $(b,s)\in\Omega_x, b\neq b_0, s\neq s_0$ are delivered by the same couriers as in $\mathbf{y}^{\backslash d_0}$. As $\x$ is welfare-optimal, 
    \begin{align*}
        v_{b_0}(s_0) - c_{d_0}(b_0,s_0) + \sum_{bs:z_{bs}=1}v_b(s) - C_{\Omega_\x}(G_D\backslash d_0) &\leq W(\mathbf{x})=\sum_{bs:z_{bs}=1}v_b(s) - C_{\Omega_\x}(G_D)\\
        v_{b_0}(s_0) - c_{d_0}(b_0,s_0) &\leq C_{\Omega_\x}(G_D\backslash d_0) - C_{\Omega_\x}(G_D) = \bar{u}_{d_0}(\x)\\
        v_{b_0}(s_0) &\leq \ut_{b_0s_0}
    \end{align*}
    This shows all edges connecting to an unmatched store $s_0$ has weight weakly smaller than zero in $G_\x$.
    
    For a store $s_0$ that is matched to a buyer $b_0$ and delivered by a courier $d_0$ in $\x$, but also valued by some buyer $b_1$, we now show that $v_{b_1}(s_0)-\ut_{b_1s_0}\leq v_{b_0}(s_0)$. Similar as the above case, one feasible allocation $\mathbf{x}'$ is to replace the buyer--store pair $(b_0,s_0)$ by $(b_1,s_0)$, and have courier $d_0$ deliver for $(b_1,s_0)$. All other orders in $\Omega_\x$ are delivered as the same way as $\mathbf{y}^\backslash d_0$, which have a total cost weakly smaller than $C_{\Omega_\x}(G_D\backslash d_0)$ because the order $(b_0, s_0)$ no longer needs to be delivered.
    \begin{align*}
        W(\mathbf{x}) &= v_{b_0}(s_0) +\sum_{bs:z_{bs}=1,b\neq b_0} v_b(s) - C_{\Omega_\x}(G_D)\\
        &\geq W(\mathbf{x}') \geq v_{b_1}(s_0) - c_{d_0}(b_1,s_0) + \sum_{bs:z_{bs}=1,b\neq b_0} v_b(s) - C_{\Omega_\x}(G_D\backslash d_0)\\
        v_{b_0}(s_0) &\geq v_{b_1}(s_0) - c_{d_0}(b_1,s_0)-\bar{u}_{d_0}(\x)=v_{b_1}(s_0)-\ut_{b_1s_0}
    \end{align*}
    This shows for any store $s_0$ already matched in $\x$, its matched edge has weakly larger weight than any edges not matched in $G_\x$. This completes the proof that $\mathbf{x}$ is the maximum weight matching in $G_\x$.
\end{proof}

\section{Zero Store Costs and Store Incentives}\label{sec:store_incentives}
Throughout the paper we made two statements on stores. First, we assume stores charge the platform its product price, which is equal to its cost. Second, we normalized the store costs to zero, and in doing so normalized the product prices to zero as well. 
In \cref{app:zero_store_cost}, we show that it is without loss to assume zero store cost and product price. In \cref{app:ce_existence}, we formulate a competitive equilibrium following the literature on multilateral contracting \citep{hatfield2011multilateral} and on competitive equilibrium in trading networks \citep{ostrovsky2008stability,hatfield2013stability}. This competitive equilibrium definition addresses store incentives, satisfies the first welfare theorem. But it is NP-hard to determine its existence. 

\subsection{Generalizing to Nonzero Store Costs}\label{app:zero_store_cost} 
All our results generalize to nonzero costs, under the first assumption that every store charges the platform its product price equaling its cost. We will illustrate the case for with-tip equilibrium. The without-tip equilibrium follows suit.

Consider a market $M'$ with store cost, where each store $s$ has cost $c_s>0$, a buyer has valuation $v'_b(s)$, and a courier has cost $c_d(b,s)$ delivering the order $(b,s)$. Welfare associated with a trade $(b,s,d)$ is $v'_b(s)-c'_s-c_d(b,s)$. The purchase price $\p'$ includes the stores' product price and the platform's delivery fee, where the product price is set to store cost. Hence, $\forall s, p'_s\geq c_s$. 
Now define the \emph{store cost version of with-tip equilibrium} the same way as the with-tip equilibrium, with the only difference that stores not bought from (i.e., $\sum_{bd}x_{bsd}=0$) have purchase price equaling its cost $p'_s=c_s$.

Given $M'$, we can construct another market $M$ without store cost. A buyer $b$ has valuation $v_b(s)=v'_b(s)-c_s$ for store $s$, and a courier has cost $c_d(b,s)$ delivering the order $(b,s)$. Welfare associated with a trade $(b,s,d)$ in $M'$ is equal  to $v_b(s)-c_d(b,s)$.

Similarly, given a market $M$ without cost, we can construct a market $M'$ with cost by setting $v'_b(s)=v_b(s)+c_s$ for some store cost $c_s$.
\begin{proposition}
    $(\p,\w,\bt,\x)$ is a with-tip equilibrium in $M$ if and only if $(\p',\w,\bt,\x)$ is a store cost version of with tip equilibrium in $M'$, where $p'_s=p_s+c_s$. 
\end{proposition}
This is very easy to check from the definition of equilibrium. Couriers incentives do not change because delivery compensation and tips remains the same. For a buyer, we can show the following two equivalence
$$\begin{cases}
    v_b(s)-p_s-\ut_{bs}\geq v_b(s')-p_{s'}-\ut_{bs'} &\Leftrightarrow v'_b(s)-p'_s-\ut_{bs}\geq v'_b(s')-p'_{s'}-\ut_{bs'}\\
    v_b(s)-p_s-\ut_{bs}\geq 0 &\Leftrightarrow v'_b(s)-p'_s-\ut_{bs}\geq 0
\end{cases}$$
So given a market $M'$ with tip, we can first solve for the with-tip equilibrium in $M$, then transfer it back to $M'$.

\subsection{Store Incentives through Competitive Equilibrium}\label{app:ce_existence}
In this section, we follow the standard definition of competitive equilibrium in the multilateral contracting literature. This competitive equilibrium definition includes stores' incentives, but does not permit tips. We write out the first and second welfare theorem, and show that a competitive equilibrium does not always exist. In fact, it is NP-hard to determine the existence of a competitive equilibrium. Since this is a separate part from the main body of paper, it is easier to redefine notations. 

Let $B,S,D$ denote the set of unit demand buyers, unit supply stores, and unit capacity couriers where $|B|=|S|=|D|=n$, denoted by $i,j,k$ respectively. The sets of all trades is denoted by $\Omega = B\times S\times D$, where $w_{ijk}$ means the trade where buyer $i$ purchase from store $j$, delivered by courier $k$. An agent $a\in B\cup S\cup D$ has valuation $v^a(w_{ijk})\in R$ for trade $w_{ijk}$ concerning this agent $a\in\{i,j,k\}$ and 0 for trades that do not concern her. $v^a(w_{ijk})$ can be positive or negative, representing a buyer's valuation, or a store or courier's cost. Each trade is associated with prices $p_\omega=(p^i_\omega,p^j_\omega,p^k_\omega)$ denoting the payments agents make, where $p^i_\omega\geq 0$ is the amount the buyer pays, $p^j_\omega\leq 0$ is the amount the store pays, or the negation of what the store receives, and $p^k_\omega\leq 0$ being the negation of what the courier receives. A pricing scheme assigns a price for all trades in $\Omega$ and is \emph{budget balanced} if $p^i_\omega+p^j_\omega+p^k_\omega=0$. An agent $a$ participating in trade $\omega$ has utility $u^a(\omega)=v^a(\omega)-p^a(\omega)$. A \textit{competitive equilibrium} $(\psi^\star,p^\star)$ is a set of trades that occur $\psi^\star\subseteq \Omega$, and a pricing scheme $p^\star$, such that (i) no agent participates in more than one trade in $\psi^\star$, (ii) the pricing scheme is budget balanced, and (iii) agents participate in their favorite trades in $\psi^\star$ at $p^\star$, with an outside option of utility zero.
As is common in the multilateral matching literature, the vital difference from competitive equilirbium(CE) in two sided is that we do not require trades that do not take place have price zero. We require budget balance instead.

\subsection{Existence of Competitive Equilibrium}
Just like CE in two sided markets, we can write out the \textit{primal integer program (PIP)}  and the \textit{dual of the linear program relaxation (DLPR)} of the central planner's welfare maximization problem.
\begin{align*}
    \text{max \;} & \sum_{ijk}x_{ijk} (v^i(w_{ijk})+v^j(w_{ijk})+v^k(w_{ijk}))\\
     & x_{ijk}\in\{0,1\} \qquad \forall i,j,k\\
     & \sum_{jk}x_{ijk}\leq 1 \quad \forall i; \quad \sum_{ik}x_{ijk}\leq 1 \quad \forall j; \quad \sum_{ij}x_{ijk}\leq 1 \quad \forall k
\end{align*}

\begin{align*}
    \text{min \;} & \sum_{i}u_i +\sum_{j}u_j+\sum_{k}u_k\\
     & u_{a}\geq 0 \qquad \forall a\in B\cup S\cup D\\
     & v^i(w_{ijk})+v^j(w_{ijk})+v^k(w_{ijk})\leq u_i+u_j+u_k \quad \forall i,j,k
\end{align*}

\begin{proposition}[First Welfare Theorem]
    The competitive equilibrium trades $\psi^\star$ maximizes welfare, even over fractional solution of the primal linear program relaxation.
\end{proposition}
\begin{proof}
    The proof is almost the same with that in two-sided markets. Let $\psi^\star_a$ be the trade that agent $a$ takes part in $\psi^\star$, and $x$ be any feasible solution to the \textit{primal linear program relaxation (PLPR)}.
    \begin{align*}
        v^i(\psi_i^\star) - p^i(\psi^\star_i) \geq \sum_{jk}x_{ijk}(v^i(\omega_{ijk}) - p^i(\omega_{ijk}))
    \end{align*}
    The results follow by summing the left and right hand side of the inequality, and noting that all prices sum to zero $\sum_{a}p^a(\psi^\star_a)=0$, and agents not trading has $v^i(\psi^\star)=0$.
\end{proof}

\begin{proposition}[Second Welfare Theorem]
    Given an integral solution $x^\star$ to PLPR, and a solution $u^\star$ to DLPR, there exists a competitive equilibrium with trades $\psi^\star$ as directed by $x^\star$, and a pricing scheme $p^\star$, where $p_a^\star(w_{ijk}) = -u^\star_a+v_a(w_{ijk})+q_{ijk}/3$ where $q_{ijk}=(u_i^\star+u_j^\star+u_k^\star-v_i(w_{ijk})-v_j(w_{ijk})-v_k(w_{ijk}))$
\end{proposition}
\begin{proof}
    It is easy to verify that prices sum to zero $p^\star_i(w_{ijk})+p^\star_j(w_{ijk})+p^\star_k(w_{ijk})=-u^\star_i-u^\star_j-u^\star_k+v_i(w_{ijk})+v_j(w_{ijk})+v_k(w_{ijk})+q_{ijk}=0$. Now by complementary slackness, if $x_{ijk}=1>0$, $q_{ijk}=0$. So $u_i^\star=v_i(w_{ijk})-p_i^\star(w_{ijk})$. For any other trades $x_{ij'k'}=0$, $u_i^\star=v_i(w_{ij'k'})-p^\star_i(w_{ij'k'})+q_{ij'k'}/3\geq v_i(w_{ij'k'})-p^\star_i(w_{ij'k'})$ so indeed the trade $(i,j,k)$ is preferred.
\end{proof}
Combinging hte first and the second welfare theorem, we have
 \begin{corollary}
     A competitive equilibrium exists if and only if the PLPR admits an integer solution. 
 \end{corollary}    
\begin{proposition}
    It is NP-hard to determine if a competitive equilibrium exists in a three-sided markets.    
\end{proposition}
\begin{proof}
    It is well known that three dimensional matching problem is NP-hard. Given a subset of hyperedges $T\subseteq V_1\times V_2\times V_3$ where $|V_1|=|V_2|=|V_3|$ and each vertex in $V_1\cup V_2\cup V_3$ has degree three, it is NP-complete to determine whether there is a perfect matching that each vertex is covered exactly once. Given an instance of such a three dimensional matching problem, construct a market where $|B|=|S|=|D|=n$ and $v_a(w_{ijk})=1$ if $a\in\{i,j,k\}$. If in polynomial time one can determine a CE exists, then an integer solution to PLPR exists. One can check if the value of PLPR equals to $3n$ and tell if a perfect matching exists. If in polynomial time one determines no CE exists, then PLPR has a fractional solution. This fractional solution cannot have value larger than $3n$ because of the PLPR constraints. So any solution to the PIP has value strictly less than $3n$. So no perfect matching exists. Thus determining if a CE exists is harder than 3DM matching.
\end{proof}

There is another version of second welfare theorem, which we append it here.
\begin{proposition}[Another Second Welfare Theorem]
    If $(\psi^\star,p^\star)$ is a competitive equilibrium, and $\phi^\star$ indicates another set of trades that is welfare efficient, then $(\phi^\star,p^\star)$ is a competitive equilibrium.
\end{proposition}
\begin{proof}
    \begin{align*}
        \sum_a u_a(\phi^\star_a,p^\star) &=\sum_{a} v_a(\phi^\star_a)-p^\star_a(\phi_a^\star)=\sum_{a} v_a(\phi^\star_a)\\
        &\geq \sum_{a} v_a(\psi^\star_a) = \sum_{a} v_a(\psi^\star_a)-p^\star_a(\psi_a) = \sum_a u_a(\psi^\star_a,p^\star)
    \end{align*}
    where the inequality is because $\phi^\star$ is welfare maximizing. But by $(\psi^\star,p^\star)$ being a CE, we know $u_a(\psi_a^\star,p^\star)\geq u_a(\phi_a^\star,p^\star),\forall a$. This with the inequalities above implies $u_a(\phi_a^\star,p^\star)= u_a(\psi_a^\star,p^\star),\forall a$.
\end{proof}
As even determining the existence of a competitive equilibrium is NP-hard, we move away from this definition. In \cref{sec:model}, we instead define an equilibrium that allows the platform to subsidize delivery, breaking the budget balance requirement. This guarantees equilibrium to exist but sacrifices the first and the second welfare theorem. 

We abstract away stores' incentives, while still keeping the stores' capacity constraints. We further require prices to be nondiscriminatory, instead of having a unique price for each trade in the definition for competitive equilibrium. With this altered notion of equilibrium, we ask how much welfare can be attained.

\end{document}